\documentclass[journal]{IEEEtran}
\usepackage{cite}
\usepackage{url}
\usepackage{amsmath}
\allowdisplaybreaks
\usepackage{amssymb}
\usepackage{enumitem}
\usepackage{bm}
\usepackage{algorithm,algorithmic}
\usepackage{mathtools}
\mathtoolsset{showonlyrefs}
\usepackage[caption=false,font=footnotesize,subrefformat=parens,labelformat=parens]{subfig}
\usepackage{color}\definecolor{RED}{rgb}{1,0,0}\definecolor{BLUE}{rgb}{0,0,1}
\DeclareMathOperator*{\minimize}{minimize}

\def\rainfty{\rightarrow\infty}
\def\ra{\rightarrow}

\def\argmin{\text{argmin}}

\def\ln{\text{ln}}

\newcommand{\mathletter}[1]{%
	\expandafter\newcommand\csname b#1\endcsname{\mathbb #1}
	\expandafter\newcommand\csname c#1\endcsname{\mathcal #1}
	\expandafter\newcommand\csname f#1\endcsname{\mathfrak #1}
	\expandafter\newcommand\csname til#1\endcsname{\widetilde #1}
	\expandafter\newcommand\csname ha#1\endcsname{\widehat #1}
	\expandafter\newcommand\csname bf#1\endcsname{\bf #1}
}%
\def\mathletters#1{\mathlettersB #1,,}
\def\mathlettersB#1,{\ifx,#1,\else\mathletter #1\expandafter\mathlettersB\fi}
\mathletters{A,B,C,D,E,F,G,H,I,J,K,L,M,N,O,P,Q,R,S,T,U,V,W,X,Y,Z}

\def \qed {\hfill \vrule height6pt width 6pt depth 0pt}
\def\bea{\begin{equation}\begin{aligned}}
\def\ena{\end{aligned}\end{equation}}
\def\bee{\begin{equation}}
\def\ene{\end{equation}}

\renewcommand{\vec}[1]{\mathbf{#1}}
\providecommand{\diff}{}

\newtheorem{theo}{Theorem}
\newtheorem{lemma}{Lemma}
\newtheorem{assum}{Assumption}

\newtheorem{remark}{Remark}

\newtheorem{example}{Example}
\newenvironment{proof}{\begin{IEEEproof}}{\end{IEEEproof}}

\def\T{\mathsf{T}}

\def\bone{{\mathbf{1}}}
\def\bzero{{\mathbf{0}}}

\begin{document}

\title{AsySPA: An Exact Asynchronous Algorithm for Convex Optimization Over Digraphs}
\author{Jiaqi~Zhang,  Keyou~You, \IEEEmembership{Senior Member,~IEEE} 
	\thanks{*This work was  supported by the National Natural Science Foundation of China under Grant 61722308 ({\em Corresponding author: Keyou You}).}
	\thanks{The authors are with the Department of Automation, and BNRist, Tsinghua University, Beijing 100084, China. E-mail: zjq16@mails.tsinghua.edu.cn, youky@tsinghua.edu.cn.}
}

\maketitle

\IEEEpeerreviewmaketitle

\begin{abstract}
	This paper proposes a novel exact distributed asynchronous subgradient-push algorithm (AsySPA) to solve an additive cost  optimization problem over directed graphs where each node only has access to a local convex function and updates asynchronously with an arbitrary rate. Specifically,  each node of a strongly connected digraph does not wait for updates from other nodes but simply starts a new update within any bounded time interval by using  local information available from its in-neighbors. ``Exact" means that every node of the AsySPA can asymptotically converge to the same optimal solution, even under different update rates among nodes and bounded communication delays. To address uneven update rates, we design a simple mechanism to adaptively adjust stepsizes per update in each node, which is substantially different from the existing works.  Then, we construct a delay-free augmented system to address asynchrony and delays, and study its convergence by proposing a generalized subgradient algorithm, which clearly has its own significance and helps us to explicitly evaluate the convergence rate of the AsySPA. Finally, we demonstrate advantages of the AsySPA in both theory and simulation.
\end{abstract}


\section{Introduction}\label{sec1}
Our objective is to distributedly solve an additive cost optimization problem  over a directed graph (digraph) where each summand is privately preserved by a computing node. This problem is usually referred to as the {\em distributed optimization problem} (DOP) and has attracted considerable attention in the past decade, see e.g. \cite{nedic2017network} and references therein.

To solve DOPs, there are two types of  distributed algorithms: \emph{synchronous} and \emph{asynchronous} algorithms. The distinct feature between them is that nodes in asynchronous algorithms do not wait for updates from other nodes but simply compute updates using its currently available information. This allows them to complete updates much faster than the synchronous ones and eliminates the costly synchronization penalty, which is especially outstanding for the case with a large scale number of heterogeneous  nodes \cite{hannah2018abcd}. However, asynchrony also brings many challenges in the design and analysis of effective algorithms. This paper proposes a novel exact asynchronous subgradient-push algorithm (AsySPA) to solve the DOP over digraphs.

\subsection{Synchronous versus asynchronous algorithms}\label{sec_1b}
We use Fig. \ref{fig3} to elaborate the striking differences between synchronous and asynchronous algorithms by taking the digraph in Fig. \ref{fig4} as an example.  In Fig. \ref{fig3},  the line with an arrow represents the communication direction  between nodes, which also allows possible transmission delays. In a synchronous algorithm  (see Fig. \subref*{fig3a:exp2}),  {\em all} the nodes start to update at the same time $t(k)$ where $k$ denotes the number of updates of the network and should be {\em known} to all nodes.  In an asynchronous algorithm (see Fig. \subref*{fig3b:exp2}), each node computes update {\em independently}. \diff{Particularly, it can start a new update immediately after finishing the current one by possibly using multiple receptions from one neighbor. }

\begin{figure}[!t]
	\centering
	\includegraphics[width=0.5\linewidth]{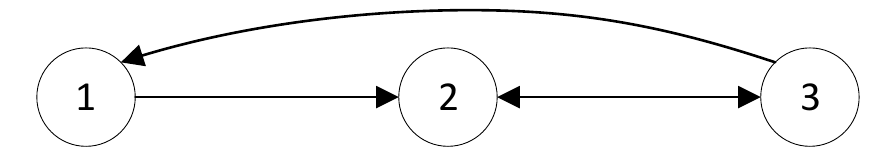}\\
	\caption{A digraph with $3$ computing nodes.}
	\label{fig4}
\end{figure}

\begin{figure}[!t]
	\centering
	\subfloat[A synchronous algorithm]{\label{fig3a:exp2}{\includegraphics[width=0.8\linewidth]{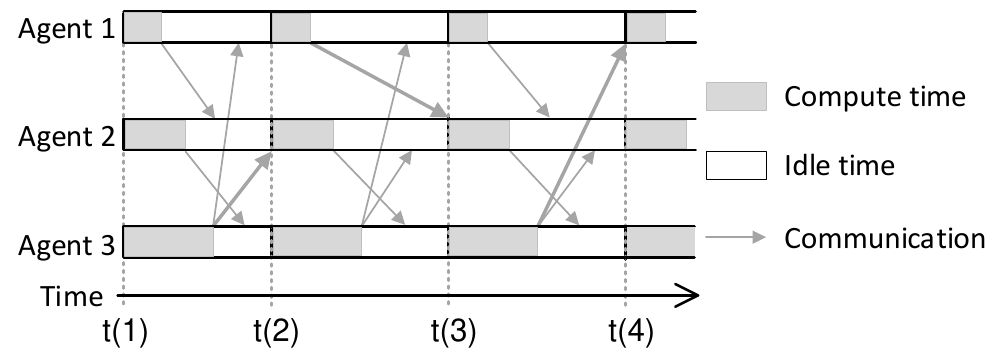}}}\\
	\subfloat[An asynchronous algorithm]{\label{fig3b:exp2}{\includegraphics[width=0.8\linewidth]{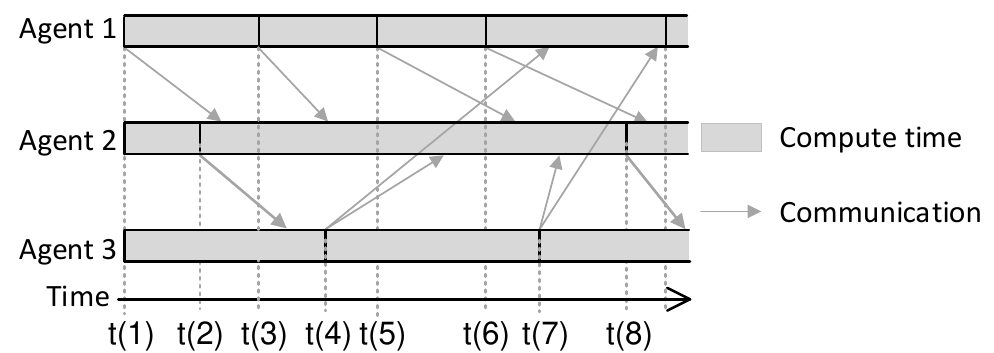}}}
	\caption{Synchronous versus asynchronous algorithms. In synchronous algorithms, all nodes start to update at the same time $t(k)$. \diff{In asynchronous algorithms, each node is allowed to compute a new update at any time without waiting for each other, and use stale information}.}
	\label{fig3}
\end{figure}

Obviously,  all the nodes in a synchronous algorithm need to agree on the update time $t(k)$, which usually needs a global clock or synchronization of all nodes. It is worthy mentioning that the clock synchronization is not easy for a large-scale system and has been studied for quite a long time   \cite{xie2018fast}. From this point of view,  an asynchronous algorithm is much easier to implement. Since some nodes may compute updates much faster, it takes relatively long idle time to wait for the slowest node in a synchronous algorithm, which potentially reduces the computational efficiency. In fact, many numerical experiments indeed suggest better performance of asynchronous algorithms \cite{assran2017empirical}. A more throughout discussion can be found in \cite{bertsekas1989parallel}.
	
However, it is much more challenging to design an effective asynchronous algorithm.  Usually, a good synchronous algorithm may become invalid if we naively  implement it in an asynchronous fashion. For example, if there exists asynchrony in implementing the celebrated synchronous SPA (SynSPA) \cite{nedic2015distributed}, each node can only converge to a neighborhood of an optimal solution, the size of which depends on the degree of asynchrony \cite{assran2018asynchronous}. In contrast, each node of the AsySPA in this paper converges to the same optimal solution.

\subsection{Literature review}
The study of synchronous algorithms for the DOP is relatively mature, see e.g. \cite{nedic2009distributed,shi2015extra,qu2017harnessing,zhang2017distributed,magnusson2017convergence,nedic2015distributed,xie2018distributed,xi2017dextra,nedic2017achieving}.  Note that many distributed algorithms under time-varying graphs (e.g. \cite{nedic2015distributed}) are synchronous as a global clock is still required.
Recently, the interest is shift to the design of asynchronous algorithms \cite{notarnicola2017asynchronous,nedic2011asynchronous,bianchi2016coordinate,farina2018asynchronous,wei20131,xu2018convergence,peng2016arock,wu2018decentralized,lian2018asynchronous,zhao2015asynchronous,bof2017newton,tian2018asy,assran2018asynchronous,tsitsiklis1986distributed,cannelli2017asynchronous,eisen2017decentralized,li1987asymptotic}. To better distinguish asynchronous algorithms, let $\cT_i$ be the sequence of update times of node $i$ and define the {\em update rate} of node $i$ in the asymptotic sense, i.e.,
\bee
	R_i=\liminf_{t\rainfty}\frac{|\cT_i\cap[0, t]|}{\sum_{j=1}^n|\cT_j\cap[0, t]|}
\ene
where $|\cdot |$ returns the cardinality of a set and $n$ is total number of computing nodes. Clearly, $R_i$ is {\em unknown} for a generic asynchronous algorithm. The  AsySPA of this work aims to {\em simultaneously} overcome the following difficulties.
\begin{enumerate}\renewcommand{\labelenumi}{\rm(\alph{enumi})}
	\item  Each node independently computes a new update without waiting for others by using its \diff{currently available} information from its in-neighbor nodes.  This  can maximumly reduce the idle time of the node and fully exploit the computational resource. Obviously, this also easily results in Uneven Update Rates (UUR) among nodes, i.e., $R_i\neq R_j$ for some $i\ne j$, and that none, one or multiple nodes can update their variables at any time.
	\item The transmission delays between nodes are time-varying but bounded. In \cite{doan2017impact,wang2015cooperative,yang2017distributed,lin2016distributed}, the authors study the effect of communication delays in the context of DOPs. However, they require a global clock for implementation and cannot deal with UUR.
	\item  The graph to model  communications between nodes is directed and strongly connected. Usually, the design of  distributed algorithms over directed graphs is  much more difficult than the undirected case, see e.g. \cite{xi2017dextra,xie2018survey}.
\end{enumerate}

Asynchronous algorithms can be categorized depending on whether they are dealing with UUR or Even Update Rates (EUR).   In \cite{notarnicola2017asynchronous,nedic2011asynchronous,bianchi2016coordinate}, the convergence of their algorithms is proved only under EUR. For example, the time interval between two consecutive updates in \cite{notarnicola2017asynchronous,nedic2011asynchronous} is an identically and independently distributed process among nodes, which obviously results in EUR.  Although it is extended  to  UUR in  \cite{farina2018asynchronous}, it requires an essential process called logic-AND for synchronization, which introduces the costly synchronization penalty and is not necessary for the AsySPA. The algorithm in \cite{bianchi2016coordinate} requires that all nodes are randomly activated with the same probability, which again results in EUR.  Although the edge-based asynchronous algorithms in \cite{wei20131,xu2018convergence}  are valid with UUR, they require two nodes to concurrently compute updates, which easily results in deadlock. This is obviously very different from our case with UUR.  Moreover, the aforementioned works \cite{notarnicola2017asynchronous,nedic2011asynchronous,farina2018asynchronous,bianchi2016coordinate,wei20131,xu2018convergence} are only applicable to undirected graphs and do not consider communication delays.

The algorithms in \cite{peng2016arock,wu2018decentralized,lian2018asynchronous} are proposed to handle the UUR problem, but they require the update rate $R_i$ to design algorithms, which is impractical in real applications, see also \cite[Remark 1]{wu2018decentralized} and \cite{tian2018asy}. For example, the major results of Theorems 3-4 in \cite{wu2018decentralized} depend on the statistics of the activation random variable of each node, which is equivalent to have access to the update rate.

The asynchronous setting of this work is mostly close to that in \cite{bof2017newton,tian2018asy,assran2018asynchronous,tsitsiklis1986distributed}. As mentioned before, the algorithm in \cite{assran2018asynchronous} is a naive extension of the SynSPA and is unable to converge to an optimal solution of the DOP under UUR.  In \cite{bof2017newton}, the asynchronous algorithm combines the Newton-Raphson method with the robust push-sum consensus. Then, its convergence depends on stepsizes and initial conditions, both of which should be carefully selected. Moreover, it only considers lossy broadcast communications rather than information delays. \diff{In \cite{tian2018asy}, the ASY-SONATA is proposed by integrating ideas of the robust push-sum consensus and the SONATA \cite{sun2016distributed}, and its exact convergence is proved if each local objective function is strongly convex with Lipschitz continuous gradient. Since  the gradient tracking is key to the ASY-SONATA, it is unclear whether it is applicable on non-smooth convex functions. Moreover, its implementation seems much more complicated than that of the AsySPA.} In the seminal work \cite{tsitsiklis1986distributed}, each node asynchronously updates a component of the decision vector, which is different from the DOP. This problem is also of high interest and keeps attracting attention \cite{cannelli2017asynchronous,eisen2017decentralized}. For example, it has been adopted for the DOP in \cite{eisen2017decentralized}. However, it is unable to converge to an exact optimal solution of the DOP. The work \cite{li1987asymptotic} provides sufficient conditions for the convergence of a class of asynchronous stochastic distributed algorithms, which are related to contraction mapping and are unclear how to extend to the DOP.

\diff{Note that the graph unbalancedness also brings challenges to the distributed algorithm design, and has been recently resolved only in synchronous algorithms by using the push-sum consensus  \cite{tsianos2012push,nedic2017achieving,xi2017dextra}, both row and column-stochastic weights \cite{xin2018linear,priolo2014distributed,cai2012average}, and the unique identifier of a node \cite{xie2018distributed}, respectively. Our AsySPA adopts  the idea of push-sum consensus  for the non-smooth convex objective function. 
}

\subsection{The paper contribution and organization}
Inspired but also motivated by the limitation of the SynSPA in \cite{nedic2015distributed}, this work proposes the AsySPA to exactly solve the DOP over digraphs under UUR among nodes and time-varying communication delays.

\diff{Firstly}, we design a novel mechanism to adaptively adjust stepsizes in each node for UUR. The key idea is that for those nodes updating slower, they increase their stepsizes in computing the associated updates. By doing this, the sum of stepsizes in each node is asymptotically of the same, which resolves the issue induced by UUR and ensures that the AsySPA is remarkably robust to the degree of asynchrony among nodes.  Clearly, this significantly advances the asynchronous  algorithm in  \cite{assran2018asynchronous}. 

\diff{Secondly}, the idea of designing a delay-free and synchronous augmented system is adopted to address asynchrony and bounded delays for the DOP, although it has been initially proposed in \cite{nedic2010convergence}  only for the consensus problem with delays. 

\diff{Thirdly}, we study the convergence of the AsySPA by proposing a generalized subgradient algorithm, which bridges the gap between subgradient methods and incremental subgradient methods \cite{bertsekas2015convex}. Then, we perform the non-asymptotic analysis of this new algorithm, which clearly has its independent significance, and show how to reformulate the AsySPA as a generalized subgradient algorithm. With this tool, we prove that each node of the AsySPA converges to the same optimal solution of the DOP and evaluate its convergence rate, which explicitly shows the faster convergence of the AsySPA over the SynSPA.  Specifically, the AsySPA with bounded communication delays converges essentially of the same speed as the SynSPA with respect to the index $k$ (see Fig. \subref*{fig3b:exp2}).  Since $k$ is increased if {\em any} node completes an update, the computing time of  the AsySPA is much less than that of the  SynSPA, see Section \ref{sec_1b}.

\diff{Finally, the advantages of the AsySPA against the current state-of-the-art are validated via training a logistic regression classifier and a support vector machine on the \emph{Covertype} dataset \cite{Dua2017UCI}.}

The rest of this paper is organized as follows. In Section \ref{sec2}, we formally describe the DOP and introduce the SynSPA. In Section \ref{sec_2c}, we present the AsySPA and compare it with the literature.  In Section \ref{sec3}, we show that all nodes of the AsySPA asymptotically achieve consensus. Then, we design a generalized subgradient method in Section \ref{sec4}, which is used for the convergence analysis of the AsySPA in Section \ref{sec5}. In Section \ref{sec6},  numerical experiments are conducted to validate our theoretical results.  Section \ref{sec7} draws some concluding remarks.

\textbf{Notation}: We use $a,\vec{a},A$, and $\cA$ to denote a scalar, vector, matrix, and set, respectively. $\vec{a}^\mathsf{T}$ and $A^\mathsf{T}$ denote the transposes of $\vec{a}$ and $A$, respectively. $[A]_{ij}$ denotes the element in row $i$ and column $j$ of $A$. $|\cA|$ denotes the cardinality  of $\cA$. $\|\cdot\|$ and $\|\cdot\|_1$ denote the $l_2$-norm and $l_1$-norm of a vector or matrix, respectively. $\lfloor x\rfloor$ denotes the largest integer less than or equal to $x$. $\bR$ denotes the set of real numbers, $\bN$ denotes the set of positive integers. $\bone_n$ and $\bzero_n$ denote the vector with all ones and all zeros, respectively. With a slight abuse of notation, $\nabla f(\vec x)$ is a subgradient of a convex function $f(\vec x)$ at $\vec x$, i.e.,
\bee\label{subgradient}
	f(\vec y)\geq f(\vec x)+(\vec y-\vec x)^\mathsf{T}\nabla f(\vec x),\ \forall \vec y\in \bR^n.
\ene

\section{Problem Formulation}\label{sec2}
In this section, we introduce some basics of a digraph and describe the distributed optimization problem (DOP) over digraphs. Then, we briefly recapitulate the synchronous subgradient-push algorithm (SynSPA) \cite{nedic2015distributed}, which is central to the design of the AsySPA.
\subsection{The distributed optimization problem}

We are interested in the DOP over a digraph $\cG=(\cV,\cE)$ of the following form\footnote{ For ease of notation, we only consider the case with a scalar decision variable $x\in\bR$. Our results can be readily extended to the vector case.}
\bee\label{original}
	\minimize_{x\in\bR}\ f(x):=\sum_{i=1}^n f_i(x)
\ene
where the local function $f_i(x)$ is only known by an individual node $i$ of $\cG$, and $n$ denotes the number of nodes, i.e.,  $n=|\cV|$.

The objective is to solve the DOP via directed interactions between nodes, which are denoted by $\cE$. That is,  the directed edge $(i,j)\in\cE$ if node $j$ directly receives information from node $i$. Let $\cN_\text{in}^i=\{j|(j,i)\in\cE\}\cup\{i\}$ denote the set of in-neighbors of node $i$, and $\cN_\text{out}^i=\{j|(i,j)\in\cE\}\cup\{i\}$ denote the set of out-neighbors of $i$. A path from node $i$ to node $j$ is a sequence of consecutively directed edges from node $i$ to node $j$. Then, $\cG$ is {\em strongly connected} if there exists a directed path between any two nodes of the digraph.

\subsection{The synchronous subgradient-push algorithm}
\begin{algorithm}[!t]
	\caption{The SynSPA}\label{alg_spa}
	\begin{itemize}[leftmargin=*]
		\item{\bf Initialization:} each node $i$ set $y_i(1)=1$ and $x_i(1)$ be an arbitrary real number $x_i(0)$.
		\item{{\bf For} each node $i\in\cV$ at each time $k\in\bN$ {\bf do}}
		      \begin{enumerate}
			      \renewcommand{\labelenumi}{\theenumi:}
			      \item Broadcast $\widetilde x_i(k)=x_i(k)/|\cN_\text{out}^i|$ and $\widetilde y_i(k)=y_i(k)/|\cN_\text{out}^i|$ to all out-neighbors of $i$.
			      \item Receive $\widetilde x_j(k)$ and $\widetilde y_j(k)$ from each in-neighbor $j\in\cN_\text{in}^i$.
			      \item Update:
			            \bea\label{eq_spa}
				            w_i(k+1)      & =\sum_{j\in\cN_\text{in}^{i}}\widetilde x_j(k),\quad y_i(k+1)=\sum_{j\in\cN_\text{in}^{i}} \widetilde y_j(k),\quad\\ z_i(k+1)&=\frac{w_i(k+1)}{y_i(k+1)}, \\
				            x_i(k+1)      & = w_i(k+1)-\rho(k)\nabla f_i(z_i(k+1)).
			            \ena
			      \item $k\leftarrow k+1$
		      \end{enumerate}
		\item{{\bf Until} a stopping criteria is satisfied}
		\item{{\bf Return} $z_i(k)$.}
	\end{itemize}
\end{algorithm}
To solve the DOP in \eqref{original}, the seminal work \cite{nedic2015distributed} proposes a novel SynSPA  (see Algorithm \ref{alg_spa}), which converges to some optimal solution of problem \eqref{original} for all $i\in\cV$ under the following condition.
\begin{assum}\label{assum}
	\begin{enumerate}
		\renewcommand{\labelenumi}{\rm(\alph{enumi})}
		\item	The digraph $\cG$ is strongly connected.
		\item The stepsize $\{\rho(k)\}$ is a nonnegative and decreasing sequence and satisfies that
		      \bee
			      \sum_{k=0}^\infty\rho(k)=\infty,\quad \sum_{k=0}^\infty\rho(k)^2<\infty.
		      \ene
		\item The local function $f_i(x)$ is convex for all $i\in\cV$, and $f(x)$ has at least one optimal solution $x^\star$, i.e., $f(x^\star)=\inf_{x\in\bR} f(x)$.
		      Moreover, there exists a $c>0$ such that\footnote{\diff{The bounded subgradient is commonly used in subgradient methods for ease of notations, see e.g. \cite{nedic2015distributed}. A weaker assumption can be found in Exercise 3.6 of \cite{bertsekas2015convex}.}}
		      \bee\label{uppergrad}
			      |\nabla f_i(x)|\leq c,\ \forall i\in\cV,x\in\bR.
		      \ene
	\end{enumerate}
\end{assum}

The SynSPA has been well appreciated  in the literature and motivated several famous algorithms, e.g., \cite{xi2017dextra,xi2018add,nedic2017achieving,xie2018distributed,zeng2015extrapush,xu2017distributed}, some of which adopt its idea to address the unbalancedness of digraphs. However, the SynSPA  and almost all variants are synchronous. For instance, all nodes in Algorithm \ref{alg_spa} compute their updates simultaneously and a global clock is required for synchronization. There is one exception \cite{assran2018asynchronous}, which {\em naively} extends the SynSPA to an asynchronous version. Unfortunately, it  only ``converges" to a neighborhood of an optimal solution of the DOP in \eqref{original}, the size of which depends on the degree of asynchrony, if all $f_i(x)$ are strongly-convex with Lipschitz-continuous gradients. As explicitly stated in \cite{assran2018asynchronous}, the inexact convergence results from the UUR among nodes, which is unavoidable in asynchronous algorithms.  This work proposes the exact AsySPA by  adaptively tuning stepsizes to resolve the UUR issue.

\section{The AsySPA}\label{sec_2c}
The AsySPA is provided in Algorithm \ref{alg_asyspa}, where we drop the iteration index $k$ to emphasize the fact of asynchronous implementation and simplify notations.

The implementation of the AsySPA is simple and as follows\footnote{\diff{As in the SynSPA, the implementation requires each node to know its out-degree, which is common in the literature \cite{hendrickx2015fundamental}.}}. Each node $i$ keeps receiving messages from its in-neighbors, and copies to their local buffers. When node $i$ is {\em locally} activated by some predefined event to compute a new update, it simply reads {\em all} the stored data in each buffer to
perform computation in \eqref{eq_asyspa}. The used data are then discarded from buffers. Whilst computing, node $i$ may also receive new data from its in-neighbors, which are stored in buffers only for the next update. \diff{Note that a node may obtain multiple receptions from one neighbor before next update, all of which are stored in buffers and used.} After completing the computation in \eqref{eq_asyspa}, the updated triple $(\widetilde x_i,\widetilde y_i,l_i)$ is broadcast to its out-neighbors, which include node $i$ as well. \diff{In practice, buffers are essentially not needed since operations over them are either taking summation or maximization in \eqref{eq_asyspa}, both of which can be done online. For example, without $\cL_i$ we can simply keep $\widetilde{l}$ and update its value immediately only if a new $l_j$ from an in-neighbor is received and $l_j>\widetilde{l}$.}

In comparison  with the SynSPA, there are some striking differences. Firstly, the activation of node $i$ is operated locally and is independent of other nodes. For example, if node $i$ observes that one of its buffers is close to full and its computation task is empty, it generates an activation to compute a new update. It can also generate an activation once an update is completed or periodically via a local time clock. While in the SynSPA, each node needs to have access to a global clock to synchronize the index $k$. \diff{Secondly, not only the latest triple $(\widetilde x_j,\widetilde y_j,l_j), j\in\cN_\text{in}^{i}$ but also the stale triples from in-neighbors are possibly used to compute an update. It should be stressed that $\cX_i$ may contain multiple receptions from a neighbor, which is used to address the asynchrony problem}. In the SynSPA, only \diff{the most recent} iterate from a node is used per update.

The last but not the least, an auxiliary integer variable $l_i$ is designed in the AsySPA, which renders it substantially different from the asynchronous algorithm in \cite{assran2018asynchronous} and is not necessary in SynSPA.  The variable $l_i$ is introduced only for adaptively adjusting stepsizes $\{\alpha_i\}$ in \eqref{eq_asyspa}, and is key to the {\em exact} convergence of the AsySPA. However, the computation cost of the AsySPA is essentially of the same as that of the SynSPA. If there is no asynchrony among nodes, the AsySPA apparently reduces to the SynSPA.

\begin{algorithm}[!t]
	\caption{The AsySPA}\label{alg_asyspa}
	\begin{itemize}[leftmargin=*]
		\item{\bf Initialization:} Each node $i$ sets $l_i=1$, $y_i=1$, assigns $x_i$ an arbitrary real number and creates local buffers $\cX_i$, $\cY_i$ and $\cL_i$. Then it broadcasts $\widetilde x_i=x_i/|\cN_\text{out}^i|$, $\widetilde y_i=y_i/|\cN_\text{out}^i|$ and $l_i$ to its out-neighbors.
		\item{{\bf For} each node $i\in\cV$, {\bf do}}
		      \begin{enumerate}
			      \renewcommand{\labelenumi}{\theenumi:}
			      \item Keep receiving $\widetilde x_j$, $\widetilde y_j$ and $l_j$ from in-neighbors and storing them respectively to the buffers $\cX_i$, $\cY_i$ and $\cL_i$, until the node is activated to update.
			      \item Compute $w_i$, $y_i$, $z_i$ and $x_i$ by
			            \bea\label{eq_asyspa}
				            w_i      & =\sum_{\widetilde x_j\in\cX_i}\widetilde x_j,\quad y_i=\sum_{\widetilde y_j\in\cY_i}\widetilde y_j,\quad z_i=\frac{w_i}{y_i}, \\
				            \widetilde l & = \max_{l_j\in\cL_i}l_j,\  \alpha_i=\sum\nolimits_{k=l_i}^{\widetilde l}\rho(k)\\
				            x_i      & = w_i-\alpha_i\nabla f_i(z_i).
			            \ena
			            and let  $l_i= \widetilde l+1$.
			      \item Broadcast $\widetilde x_i=x_i/|\cN_\text{out}^i|$, $\widetilde y_i=y_i/|\cN_\text{out}^i|$ and $l_i$ to all out-neighbors of node $i$, and empty $\cX_i$, $\cY_i$ and $\cL_i$.
		      \end{enumerate}
		\item{{\bf Until} a stopping criteria is satisfied}
		\item{{\bf Return} $z_i$.}
	\end{itemize}
\end{algorithm}

Without adaptively adjusting stepsizes, it is just the asynchronous algorithm in \cite{assran2018asynchronous}, which can only converge to a neighborhood of an optimal solution of the DOP in \eqref{original}. The size of the neighborhood depends on the degree of asynchrony. It is well acknowledged that the major difficulty in an asynchronous algorithm is how to effectively address UUR issue among nodes.  Intuitively, a local function $f_i(x)$ of the node with a high update rate  will be considered more often than the rest. If we do not adaptively adjust stepsizes (see \cite{assran2018asynchronous}), the resulting algorithm then minimizes a weighted sum of $f_i(x)$, whose weights are proportional to update rates of associated nodes. \diff{To informally exposit it, consider a special example that two nodes form an {\em undirected} graph with edge weight $1/2$ and  assume that node $1$ updates at each discrete-times, and node $2$ only updates at even times. We further let $\rho_k=\rho$, where $\rho>0$ is sufficiently small. Without adaptively tuning stepsizes, the asynchronous version of Algorithm \ref{alg_spa} can be reduced to that
\bea
&x_1(2k+2)=x_1(2k+1)-\rho \nabla f_1(x_1(2k+1))\\
&=\frac{x_1(2k)+x_2(2k)}{2}-\rho \nabla f_1(x_1(2k))-\rho \nabla f_1(x_1(2k+1))\\
&\approx 0.5x_1(2k)+0.5x_2(2k)- 2 \rho \nabla f_1(x_1(2k)),\\
&x_2(2k+2)=0.5x_1(2k)+0.5x_2(2k)-  \rho \nabla f_2(x_2(2k)),
\ena
where  $x_1(2k) \approx x_1(2k+1)$ is used as $\rho$ is sufficiently small.  In view of the distributed algorithm in \cite{nedic2009distributed}, one can easily observe that the above tends to  minimize the objective function $2 f_1(x)+f_2(x)$. Clearly, this is not the additive objective function in \eqref{original}. This phenomenon is also observed  but unresolved in \cite{assran2018asynchronous}.}

To address it, our novel idea is to adaptively adjust stepsizes by using relatively ``large" stepsizes for nodes with low update rates, and eventually the sum of stepsizes associated with each local function is approximately equal.  Specifically, if a node $i$ finds it progress less towards its associated subgradient direction than its in-neighbors, it applies a relatively larger stepsize in the next update. The role of the auxiliary variable $l_i$ in the AsySPA is thus designed to roughly record the maximum number of updates in its in-neighbors' nodes. \diff{In fact, let $\Omega_{i,t}$ be the set of stepsizes used in node $i$ before time $t$. Under mild conditions, we shall show that  $\text{sum}(\Omega_{i,t})- \text{sum}(\Omega_{j,t})$ asymptotically converges to 0 as $t$ goes to infinity  where $\text{sum}(\cX)$ returns the sum of elements in the set $\cX$}. That is, the sum of  stepsizes used in each node is eventually of the same, which is critical to the exact convergence of the AsySPA and is the key difference from \cite{assran2018asynchronous}.

\section{Consensus achieving of the AsySPA}\label{sec3}
To elaborate the exact convergence of the AsySPA, i.e., each node of the AsySPA converges to the same optimal solution of the DOP in \eqref{original}, we first show that the nodes asymptotically reach consensus by adopting the augmented system approach. Then, we prove that the average state of nodes converges to an optimal solution. To this end, two assumptions are made.

\begin{assum}[Bounded activation time interval]\label{assum3}Let $t_i$ and $t_i^+$ be  two consecutive activation time of node $i$,  there exist two positive constants $\underline{\tau}>0$ and $\bar{\tau}<\infty$ such that $\underline{\tau}\leq |t_i-t_i^+| \leq \bar{\tau}$ for all $i\in\cV$.
\end{assum}

Clearly, Assumption \ref{assum3} is easy to satisfy in implementing the AsySPA. For example, $|t_i-t_i^+|$ is bounded below naturally since the computation of each update consumes time, and thus the time interval between two updates cannot be arbitrarily small. While the computational complexity in each update is essentially of the same, each update can be computed in a finite time, and the time interval between two updates can be made finite.  \begin{assum}[Bounded transmission delays]\label{assum6}For any $(i,j)\in\cE$, the time-varying transmission delay from node $i$ to node $j$ is uniformly bounded by a positive constant $\tau>0$.
\end{assum}

Bounded transmission delays are common and reasonable \cite{lin2016distributed,doan2017impact}. Very recently, it is relaxed to an unbounded case in \cite{wu2018decentralized}. However, their analysis is limited to {\em undirected} graphs and is unclear how to extend to directed graphs. It should be noted that all the above constants are not needed in implementing the AsySPA. 

This section is dedicated to proving consensus among nodes. Particularly, we show that $z_i$ in the AsySPA asymptotically converges to the same point under the aforementioned assumptions.

\subsection{A key technical lemma}
Under Assumption \ref{assum3}, the activation time of nodes cannot be continuum. Thus, let $\cT=\{t(k)\}_{k\ge 1}$ be an increasing sequence of activation times of all nodes, e.g. $t\in\cT$ if there is at least one node being activated at time $t$. Although the index $k$ is a {\em global} counter, it is only introduced for convergence analysis and is not needed for implementing the AsySPA.
We further denote $\cT_i\subseteq\cT$  the sequence of activation times of node $i$, i.e., $t\in\cT_i$ if  node $i$ is activated at time $t$.
The following lemma is key to the convergence analysis of the AsySPA.
\begin{lemma}\label{lemma1} The following statements hold.
	\begin{enumerate}[label=(\alph*)]
		\item Under Assumption \ref{assum3}, let $b_1=(n-1)\lfloor\bar{\tau}/\underline{\tau}\rfloor+1$,  each node is activated at least once within the time interval $(t(k),t(k+b_1)]$. Moreover, let $\cG(k)=(\cV,\cE(k))$ be the effective graph at time $t(k)$, where $(i,j)\in\cE(k)$ if $t(k)\in\cT_i$ and $(i,j)\in\cE$, then the union of graphs $\bigcup_{t=k}^{k+b_1}\cG(t)$ is strongly connected for any $k$.
		\item Under Assumptions \ref{assum3} and \ref{assum6}, let $b_2=n\lfloor\tau/\underline{\tau}\rfloor$ and $b=b_1+b_2$, the information sent from node $i$ at time $t(k)$ can be received by node $j$ before time $t(k+b_2)$ and used for computing an update before time $t(k+b)$ for any $k$ and $(i,j)\in\cE$.
		\item Under Assumptions \ref{assum}(a), \ref{assum3} and \ref{assum6}, we have $|l_i(k)-l_j(k)|\leq nb$ and $0\leq l_i(k+1)-l_i(k)\leq nb+1$ for any $i,j\in\cV$ and $k$, where $l_i(k)$ is the value of $l_i$ at time $t(k)$.
	\end{enumerate}
\end{lemma}
\begin{proof} (a) Suppose that node $i$ is not activated during the time interval $(t(p),t(q)],p,q\in\bN$ but is activated at $t(q+1)$. It follows from  Assumption \ref{assum3} that $t(q)-t(p)\leq\bar{\tau}$. Moreover, any other node can be activated at most $\lfloor(t(q)-t(p))/\underline{\tau}\rfloor\leq\lfloor\bar{\tau}/\underline{\tau}\rfloor$ times during the time interval $(t(p),t(q)]$, which implies $q-p\leq (n-1)\lfloor\bar{\tau}/\underline{\tau}\rfloor$. Hence the first part of the result follows. This also implies the uniformly joint connectivity of $\cG(k)$.

	(b) Suppose that node $i$ sends information at time $t(p),p\in\bN$ and node $j$ receives it in the time interval $(t(q),t(q+1)],q\in\bN$. It follows from Assumption \ref{assum6} that $t(q)-t(p)\leq\tau$. Moreover,   Assumption \ref{assum3} implies that any node can be activated at most $\lfloor{\tau}/\underline{\tau}\rfloor$ times during the time interval $[t(p),t(q)]$, i.e., $q-p+1\leq n\lfloor {\tau}/\underline{\tau}\ \rfloor$, and hence $q+1\leq p+n\lfloor{\tau}/\underline{\tau}\rfloor$. The result follows by letting $p=k$. Jointly with Lemma \ref{lemma1}(a), the rest of the results follow immediately.

	(c) Let $\bar l(k)=\max_{i\in\cV}l_i(k)$.  Since $l_i(k+1)\leq \bar l(k)+1$ for all $i$, we have that $\bar l(k+1)\leq \bar l(k)+1$. Lemma \ref{lemma1}(b) and Assumption \ref{assum}(a) jointly imply that the information from node $i$ at time $t(k)$ must reach node $j$ before time $t(k+nb)$ for any $i,j\in\cV$. Thus, $l_i(k+nb+1)\geq \bar l(k)+1$ for all $i\in\cV$. Combining the above yields that
	\bea
		\bar l(k)+1&\leq l_i(k+nb+1)\leq \bar l(k+nb+1)\\
		&\leq \bar l(k)+nb+1
	\ena
	for all $i$ and $k$. This also implies that $|l_i(k)-l_j(k)|\leq |\bar l(k)+nb+1-\bar l(k)-1|= nb$ for all $i,j\in\cV$, and $l_i(k+1)-l_i(k)\leq \bar l(k)+1-l_i(k)\leq nb+1$.
\end{proof}

\subsection{Reformulation of the AsySPA via an augmented graph}

Denote the latest state of node $i$ just before time $t(k)$ by $w_i(k),x_i(k),y_i(k),z_i(k)$ and $l_i(k)$. Then, it is clear that
\bea\label{reforiter}
	w_i(k+1)&=w_i(k),x_i(k+1)=x_i(k),y_i(k+1)=y_i(k),\\
	z_i(k+1)&=z_i(k),l_i(k+1)=l_i(k),\ \forall t(k)\notin \cT_i.
\ena

We turn to reformulate the above iteration by designing an augmented graph to handle bounded time-varying transmission delays and asynchrony. Note that this approach has been adopted in \cite{nedic2010convergence} for studying the consensus problem with bounded transmission delays.  To this end, we associate each node $i$ with $b$ virtual nodes $\{v_i^{(1)},...,v_i^{(b)}\}$, where $b$ is defined in Lemma \ref{lemma1}(b). Hence, the total number of virtual nodes  is $nb$. \diff{Then, we construct an augmented graph\footnote{The augmented graph is also only introduced for the convergence analysis of the AsySPA and is not needed for its implementation.} $\widetilde \cG(k)=(\widetilde{\cV},\widetilde \cE(k))$ to model the communication digraphs between nodes at time $t(k)$. If $(i,j)\in\cE$ in $\cG$, then edges $(v_j^{(1)},j)$, $(v_j^{(2)},v_j^{(1)})$,..., and $(v_j^{(b)},v_j^{(b-1)})$ are always included in $\widetilde \cE(k)$, and only one of edges $(i,v_j^{(1)})$, $(i,v_j^{(2)}),\ldots, (i,v_j^{(b)})$ and $(i,j)$ shall be possibly included in $\widetilde \cE(k)$, which depends on the transmission delay or asynchrony  between node $i$ and node $j$ at time $t(k)$. Specifically, suppose that node $i$ sends message to node $j$ at time $t(k)$, which may be subject to transmission delay, and node $j$ only uses this message at time $t(k+u+1)$, where $1\le u\le b-1$ by Lemma \ref{lemma1}(b), then $(i,v_j^{(u)})\in\widetilde{\cE}(k)$. If there is no communication delay between node $i$ and node $j$ at time $t(k)$ and the transmitted message is used by node $j$ at time $t(k+1)$, then only $(i,j)\in\widetilde{\cE}(k)$. See Fig. \ref{fig1} for an augmented graph of  Fig. \ref{fig4} with virtual nodes.
\begin{figure}[!t]
	\centering
	\includegraphics[width=0.7\linewidth]{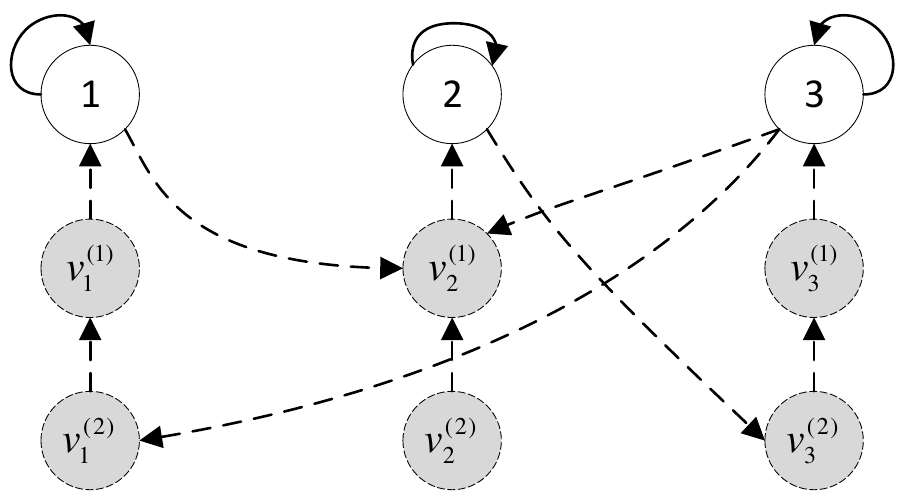}
	\caption{\diff{An augmented graph of Fig. \ref{fig4} at time $t(k)$, where a self-loop is included for each non-virtual node and $b=2$. At time $t(k)$, node $1$ sends message to node $2$.  If the message is received at the interval $(t(k+1), t(k+2)]$ due to the transmission delay and is used by node $2$ at $t(k+2)$, then only $(1, v_2^{(1)})\in \widetilde \cE(k)$. Even if node $2$ sends message to node $3$ at time $t(k)$ without delay and this message is  used by node $3$ at time $t(k+3)$ due to the asynchrony, then only $(2, v_3^{(2)})\in \widetilde \cE(k)$. Similarly for node $3$.}}
	\label{fig1}
\end{figure}

Now, we briefly explain how to use the augmented graph to address delays and asynchrony. Suppose that node $i$ sends message to node $j$ at $t(k)$ and  due to delay node $j$ receives it at time interval $(t(k+1),t(k+2)]$. However, this message is used by node $j$ at time $t(k+3)$ due to asynchrony.  In this case, all the edges $(i,v_j^{(2)}), (v_j^{(2)},v_j^{(1)}), (v_j^{(1)},j)$ are included in $\widetilde \cE(k)$. In the augmented graph, this can be viewed as that node $i$ sends message  to node $j$ via the relay nodes $v_j^{(2)}$ and $v_j^{(1)}$, respectively. Clearly, all non-virtual nodes in $\widetilde \cG$ receive the same information as that in $\cG$ and hence their updates appear to be synchronous and without delays.}

In $\widetilde \cG$, we enumerate non-virtual nodes first and then  virtual nodes. Specifically, let $\cV^{(u)}=\{v_i^{(u)},i\in\cV\}$ for all $u=\{1,...,b\}$ and $\widetilde \cV=\{\cV,\cV^{(1)},...,\cV^{(b)}\}$. That is, node $i$ of $\cG$ is the $i$-th node in $\widetilde \cG$, and the $(nu+i)$-th node in $\widetilde \cG$ is the virtual node $v_i^{(u)}$ for all $i\in\cV$ and $u\in\{1,...,b\}$. Let $\widetilde n:=|\widetilde \cV|=n(b+1)$ be the number of nodes of $\widetilde \cG$. Then, the AsySPA can be expressed as a synchronous and delay-free augmented system
\bea\label{eq1_sec2}
	\widetilde{\vec x}(k+1)&=\widetilde A(k)\widetilde{\vec x}(k)-\vec g(k),\\
	\widetilde{\vec y}(k+1)& = \widetilde A(k)\widetilde{\vec y}(k),\\
	z_i(k+1) &= \frac{[\widetilde A(k)\widetilde{\vec x}(k)]_i}{\widetilde y_i(k+1)},\ \forall i\in\cV
\ena
where $\widetilde{\vec x}(k)=[\widetilde x_1(k),...,\widetilde x_{\widetilde n}(k)]^\T,\widetilde{ \vec y}(k)=[\widetilde y_1(k),...,\widetilde y_{\widetilde n}(k)]^\T$, and $\widetilde{\vec x}(1)=[\widetilde x_1(0),...,\widetilde x_{n}(0),\bzero^\T]^\T$,  $\widetilde{\vec y}(1)=[\bone_n^\T,\bzero^\T]^\T$,
\bee
	\begin{split}
		&[\widetilde A(k)]_{ji}\\
		&=\left\{\begin{array}{ll}
			{1}/{|\cN_\text{out}^i|}, & \text{if $i\in\cV$, $t(k)\in\cT_i$, $j=nu+v$, $\tau_{iv}(k)=u$} \\
			1,                            & \text{if $i\in\cV$, $t(k)\notin\cT_i$ and $j=i$}                \\
			1,                            & \text{if $i\notin\cV$ and $j=i-n$}                              \\
			0,                            & \text{otherwise}
		\end{array}\right.
	\end{split}
\ene
and $\tau_{iv}(k)\in\bN$ denotes the transmission delay from node $i$  to node $v$ at time $t(k)$. Under Assumption \ref{assum6}, $\tau_{ii}(k)=0$ for all $i\in\cV,k\in\bN$, and  $\tau_{iv}(k)=\infty$ for any $(i,v)\notin\cE$. Finally,
\bee\label{eq_g}
	\vec g(k)=[g_1(k),...,g_n(k),\bzero^\T]^\T
\ene
where
\bee
	g_i(k)= \left\{\begin{array}{ll}
		\sum_{k=l_i(k)}^{l_i(k+1)-1}\rho(k)\nabla f_i(z_i(k+1)), & \text{ if $t(k)\in\cT_i$} \\
		0,                                                       & \text{ otherwise}
	\end{array}\right.
\ene
for all $i\in\cV$.

Note that $\widetilde A(k)$ is a column-stochastic matrix. Moreover, the last equality in \eqref{eq1_sec2} is only for nodes in $\cV$, and is well-defined because $[\widetilde A(k)]_{ii}>0$ and hence $\widetilde y_i(k)>0$ for all $k$ and $i\in\cV$. In fact, the infimum of $\widetilde y_i(k)$ is strictly greater than $0$ as proved in the next subsection. Since $\lim_{k\rightarrow\infty}g_i(k)=0$ for all $i$ from Assumption \ref{assum}  and Lemma \ref{lemma1}(c), we have $\lim_{k\rightarrow\infty} \vec g(k)=\bzero$, which is key to consensus achieving.

\subsection{Asymptotic consensus among nodes}
We now show that each node of the AsySPA asymptotically converges to consensus. To this end, we define
\bee
	\Phi(k,t)=\widetilde A(k)\widetilde A(k-1)\cdots\widetilde A(t+1)\widetilde A(t)
\ene
and
\bee
	\bar{x}(k)=\frac{\bone^\T\widetilde{ \vec x}(k)}{\widetilde n}=\frac{\bone^\T\widetilde{ \vec x}(k)}{n(b+1)}
\ene
which is the average of elements in $\widetilde{ \vec x}(k)$.

The following lemma states that $\Phi(k,t)$ converges to a rank-one matrix and implies that $\widetilde y_i(k)$ is strictly greater than 0 for all $k$ and $i\in\cV$.
\begin{lemma}\label{lemma2}
	Under Assumptions \ref{assum}, \ref{assum3} and \ref{assum6}, the following statements are in force.
	\begin{enumerate}[label=(\alph*)]
		\item There exists a nonnegative vector $\phi(k)$ satisfying that $\bone^\T\phi(k)=1$ and
		      \begin{equation}
			      \|\Phi(k,t)- \phi(k)\bone^\T\|_1\leq\alpha\lambda^{k-t}
		      \end{equation}
		      for all $k,t\in\bN,k\geq t$, where
		      \bee
			      \alpha = 4n(1+n^{nb}),\ \lambda=\left(1-\frac{1}{n^{nb}}\right)^{\frac{1}{nb}}
		      \ene
		      and $b$ is defined in Lemma \ref{lemma1}(c).
		\item Moreover, $
			      \sum_{j=1}^n[\Phi(k,1)]_{ij}\geq {n^{-nb}},\ \forall i\in\cV,k\in\bN.$
	\end{enumerate}
\end{lemma}
\begin{proof}
	The first part can be found in \cite[Lemma 5]{nedic2010convergence} by jointly using  Lemma \ref{lemma1}.

	To prove (b), two cases are separately studied. If $k\leq nb$, then $[\Phi(k,t)]_{ii}\geq[\widetilde A(k)]_{ii}[\widetilde A(k-1)]_{ii}\cdots[\widetilde A(t)]_{ii}\geq (1/|\cN_\text{out}^i|)^{nb}\geq  1/n^{nb}$, and hence the result is obtained.

	If $k> nb$, it follows from \cite[Lemma 2(b)]{nedic2010convergence} that
	$[\Phi(k,k-nb+1)]_{ij}\geq 1/n^{nb}$ for all $i\in\cV$ and $j\in\widetilde\cV$. Then,
	\begin{align}
		 & [\Phi(k,k-nb)]_{ij}                                                                                                                \\
		 & =\sum_{u=1}^{\widetilde n}[\Phi(k,k-nb+1)]_{iu}[\widetilde A(k-nb)]_{uj}                                                           \\
		 & \geq \frac{1}{n^{nb}}\sum_{u=1}^{\widetilde n}[\widetilde A(k-nb)]_{uj}\geq \frac{1}{n^{nb}}.\ \forall i\in\cV,j\in\widetilde \cV.
	\end{align}
	where the last inequality follows from the column-stochasticity of $\widetilde A(k)$. The desired result is obtained by induction.
\end{proof}

The main result of this section ensures that all nodes of the AsySPA eventually achieve consensus.
\begin{lemma}\label{lemma3}
	Under Assumptions \ref{assum}, \ref{assum3} and \ref{assum6}, let $\{z_i(k)\},i\in\cV$ be generated by the AsySPA.
	\begin{enumerate}[label=(\alph*)]
		\item For all $k\geq 1$, it holds that
		      \bee
			      \hspace{-10pt}|z_i(k+1)-\bar{x}(k)|\leq 8n^{nb}\left(\lambda^k\|\vec x(1)\|_1+\sum_{t=1}^{k}\lambda^{k-t}\|\vec g(t)\|_1\right)
		      \ene
		      where $\lambda<1$ is given in Lemma \ref{lemma2} \diff{and $\vec g(k)$ is defined in \eqref{eq_g}}. In particular, we have
		      \bee
			      \lim_{k\rightarrow\infty}|z_i(k+1)-\bar{x}(k)|=0,\ \forall i\in\cV.
		      \ene
		\item Moreover,
		      $
			      \sum_{t=1}^{\infty}\rho(t)|z_i(t+1)-\bar{x}(t)|<\infty,\ \forall i\in\cV.
		      $
	\end{enumerate}
\end{lemma}
\begin{proof}
	The proof is similar to that in \cite[Lemma 1]{nedic2015distributed} and omitted for saving space.
\end{proof}

Lemma \ref{lemma3} shows that $z_i(k)$ generated by the AsySPA converges to $\bar x(k)$ for all $i$. Then, it is sufficient to prove that $\bar x(k)$ converges to an optimal solution of the DOP in \eqref{original}, which is the main focus of the remaining sections.

\diff{
\begin{remark}
	The convergence rate of $z_i(k)$ is a worst case bound, which has very bad scaling with respect to the number of nodes and  is unfortunately achievable, see \cite{nedic2016distributed} for the existence of such a graph.
\end{remark}}

\section{Generalized Subgradient Methods}\label{sec4}
This section proposes a generalized subgradient algorithm which unifies both subgradient methods and incremental subgradient methods \cite{{bertsekas2015convex}}.  Then, we study its convergence, the results of which not only have its independent significance but also are essential for proving the exact convergence of the AsySPA in the next section.

We still focus on the optimization problem in \eqref{original} where the decision vector $\vec x \in\bR^m $ is now multi-dimensional. Given two sequences $\cS=\{s(1),s(2),...\}\subseteq\cV$ and $\cR_i=\{r_i(1),r_i(2),...\}\subseteq\bN\cup\{0\}$ which satisfy that
\bee\label{eq1_sec4}
	r_i(k+1)=\left\{\begin{array}{ll}
		r_i(k)+\Delta r_i(k), & \text{if}~s(k)=i, \\
		r_i(k),               & \text{otherwise},
	\end{array}\right.
\ene
where $\Delta r_i(k)\in\bN$ for all $i\in\cV$.

Now, we propose the following {\em generalized subgradient method} to solve the optimization problem in \eqref{original}
\bee\label{ism}
	\vec x(k+1)=\vec x(k) -\left(\sum_{t=r_{s(k)}(k)+1}^{r_{s(k)}(k+1)}\rho(t)
	\right)\nabla f_{s(k)}(\vec x(k)+\varepsilon(k))
\ene
where  $\varepsilon(k)\in\bR^n$ denotes noise or computing errors. The novelty of \eqref{ism} lies in the use of adaptive stepsizes. By choosing different sequences $\cS$ and $\cR_i$, then  \eqref{ism} covers  a broad class of subgradient methods, e.g., the subgradient method and the incremental subgradient method as elaborated below.

\begin{example}[Incremental subgradient method] Let $\cS=\{1,2,...,n,1,2,...,n,...\}$, $\cR_i=\{\bzero_i^\T,\bone_n^\T,2\bone_n^\T,3\bone_n^\T,...\}$ for all $i$, and $\varepsilon(k)=0$ for all $k$. Then, \eqref{ism} is rewritten as
	\bee\label{ism2}
		\vec x(k+1)=\vec x(k) - \rho(t)\nabla f_p(\vec x(k)),
	\ene
	where $t=\lfloor (k-1)/n \rfloor+1$ and $p=k-(t-1)n$. Clearly, \eqref{ism2} reduces to the cyclic incremental subgradient method (cf. \cite[Section 2.1]{bertsekas2015convex}). The randomized incremental subgradient method can also be derived in a similar way.\qed
\end{example}
\begin{example}[Subgradient method] Let $\cS$ and $\cR_i$ be the same as in Example 1 but
	\bee\label{eq2_sec4}
		\varepsilon(k)=\vec x(n\lfloor \frac{k}{n} \rfloor+1 )-\vec x(k)
	\ene
	for all $k$. It follows from \eqref{ism} that
	\bee
		\|\varepsilon(k)\|\leq n\rho(\lfloor \frac{k-1}{n} \rfloor+1)\max_{i}\|\nabla f_i(\vec x(n\lfloor \frac{k}{n} \rfloor+1))\|.
	\ene

	Under Assumption \ref{assum}, it implies that $\lim_{k\rightarrow\infty}\|\varepsilon(k)\|=0$. Substituting \eqref{eq2_sec4} into \eqref{ism} leads to that
	\bee\label{gd}
		\vec x(k+n)=\vec x(k) - \rho(t)\nabla f(\vec x(k)),
	\ene
	where $\nabla f(\vec x(k))=[\nabla f_1(\vec x(k)),...,\nabla f_n(\vec x(k))]^\T$ and $t=\lfloor (k-1)/n \rfloor+1$. Let $\vec y(k)=\vec x(kn-n+1)$, then \eqref{gd} can be written as
	\bee\label{gd2}
		\vec y(k+1)=\vec y(k) - \rho(k)\nabla f(\vec y(k)).
	\ene
Then, it is a  subgradient method with stepsize $\rho(k)$.\qed
\end{example}

Now, we study the convergence of \eqref{ism} under the following assumption on the sequences $\cS$, $\cR_i, i\in\cV$ and $\{\varepsilon(k)\}$.
\begin{assum}\label{assum5}
	\begin{enumerate}[label=(\alph*)]
		\item There exists a $\sigma_1\in\bN$ such that $\cV\subseteq\{s(k+1),...,s(k+\sigma_1)\}$ for all $k\in\bN$.
		\item There exists a $\sigma_2\in\bN$ such that $|r_i(k)-r_j(k)|\leq \sigma_2$ and $\Delta r_i(k)\leq \sigma_2$ for all $i,j\in\cV$ and $k\in\bN$.
		\item $\sum_{k=1}^{\infty}\rho(k)\|\varepsilon(k)\|<\infty$.
	\end{enumerate}
\end{assum}

Assumption \ref{assum5}(a) guarantees that each function $f_i(x)$ in \eqref{original} needs  to be updated at least once in every $\sigma_1$ steps. Otherwise, $\vec x(k)$ will be significantly affected by some local function(s), and cannot  converge to an optimal solution.  Assumption \ref{assum5}(b) ensures that the stepsizes for updating each function $f_i(x)$ should not be  overwhelmingly unbalanced. Assumption \ref{assum5}(c) requires the cumulative noise to be bounded. Overall, Assumption \ref{assum5} is reasonable and even necessary in some cases. In particular, one can easily find counterexamples where $\{\vec x(k)\}$ in \eqref{ism} diverges  if any condition in Assumption \ref{assum5} is violated. Note that $\cS$, $\cR_i, i\in\cV$ and $\{\varepsilon(k)\}$ in Examples 1-2 satisfy Assumption \ref{assum5} if Assumption \ref{assum} holds.
\begin{theo}\label{prop2}
	Suppose that Assumptions \ref{assum} and \ref{assum5} hold. Then, $\{\vec x(k)\}$ generated by \eqref{ism} converges to an optimal solution of the problem \eqref{original}.
\end{theo}
\begin{proof}
	See Appendix \ref{appendix_a}.
\end{proof}

\begin{remark}
	Due to the space limitation, we only consider diminishing stepsizes  in \eqref{ism}. As in the standard subgradient methods, constant stepsizes can only ensure that $\vec x(k)$ is eventually attracted to a neighborhood of an optimal solution.
\end{remark}

To evaluate the convergence rate of \eqref{ism}, define
\bee
	d(\vec x)=\inf_{\vec x^\star\in\vec \cX^\star}\|\vec x-\vec x^\star\|
\ene
and make the following assumption.
\begin{assum}[H\"{o}lder Error Bound \cite{johnstone2017faster}]\label{assum8}
	The objective function $f(\vec x)$ in \eqref{original} satisfies
	\bee
		f(\vec x) - f^\star\geq c_hd(\vec x)^{1/\theta},\ \forall \vec x\in\bR^n
	\ene
	for some $c_h>0$ and $\theta\in(0,1]$.
\end{assum}

The H\"{o}lder error bound characterizes the order of growth of $f(\vec x)$ around an optimal solution, and holds in many applications, see \cite{johnstone2017faster} for details.
\begin{theo}\label{prop3}Let $\{x(k)\}$ be generated by \eqref{ism} where $\rho(k)=k^{-\alpha}$.
	Suppose that Assumptions \ref{assum}(c) and  \ref{assum5} hold, and $\|\varepsilon(k)\|\leq c_{\varepsilon}\rho(k)$ for some $c_{\varepsilon}>0$. Denote $\sigma=\sigma_1+\sigma_2$, the following statements hold.
	\begin{enumerate}[label=(\alph*)]
		\item If $\alpha\in(0.5,1]$, we have
		      \bea
			      &\min_{1\leq t\leq k\sigma}f(\vec x(t))-f^\star\leq\\
			      &\quad\frac{4\sigma^3c}{s(k)}\left(\frac{2(c_{\varepsilon}+3\sigma^2c)}{2\alpha-1}+3\sigma^2cd(\vec x(1))^2\right)
		      \ena
		      where
		      \bee\label{conrate}
			      s({k})=\left\{\begin{array}{ll}\frac{1}{1-\alpha}(k^{1-\alpha}-1),&\text{ if }\alpha\in(0.5,1),\\ \ln(k),&\text{ if }\alpha=1.\end{array}\right.
		      \ene
		      If $\alpha=0.5$, we have
		      \bea
			      &\min_{1\leq t\leq k\sigma}f(\vec x(t))-f^\star\\
			      &\leq\frac{4\sigma^3c}{\sqrt{k}}\left(2(c_{\varepsilon}+3\sigma^2c)\ln(k)+3\sigma^2cd(\vec x(1))^2\right).
		      \ena
		\item Suppose that Assumption \ref{assum8} also holds with $\theta\in[0.5,1]$. Let $\alpha\in(0,1]$,  there exists some $0<c_0<\infty$ such that
		      \bee
			      d(\vec x(k))^2\leq c_0 k^{-2\theta\alpha}
		      \ene
		      where $c_0$ depends on $\sigma_1$ and $\sigma_2$ in Assumption \ref{assum5}.
	\end{enumerate}
\end{theo}
\begin{proof}
	See Appendix \ref{appendix_b}.
\end{proof}

Theorem \ref{prop3}(a) reveals that $f(\vec{x}^k)$ in \eqref{ism} \diff{can achieve a convergence rate of $O(\ln(k)/\sqrt{k})$, which is consistent with standard subgradient methods \cite[Section 3.2.3]{nesterov2013introductory}}. Theorem \ref{prop3}(b) cannot be directly obtained by using Theorem \ref{prop3}(a), which only implies that $d(\vec x^k)^2$ converges at the best rate of $O(k^{-\theta})$. However, Theorem \ref{prop3}(b) dictates that $d(\vec x^k)^2$ can actually converge at a rate of $O(k^{-2\theta})$ by using stepsizes $\rho(k)=1/k$.

\section{Exact Convergence of the AsySPA}\label{sec5}

In this section, we prove the exact convergence and evaluate the convergence rate of the AsySPA.
\begin{theo}\label{theo1}
	Suppose that Assumptions \ref{assum}-\ref{assum6} hold. Then, each node of the AsySPA asymptotically converges to the same optimal solution $x^\star\in\cX^\star$, i.e., $\lim_{k\rainfty}z_i(k)=x^\star$ for all $i\in\cV$ where $z_i(k)$ is defined in \eqref{eq1_sec2}.
\end{theo}

\begin{proof}
	In view of Lemma \ref{lemma3}, it is sufficient to show that $\lim_{k\rainfty}\bar x(k)=x^\star$ and follows that
	\bee
		z_i(k+1)=\bar x(k)+\epsilon_i(k)
	\ene
	where $\lim_{k\rightarrow\infty}\epsilon_i(k)=0$ and $\sum_{k=1}^{\infty}\rho(k)|\epsilon_i(k)|<\infty$.

	By left-multiplying $\bone^\T/n$ on both sides of \eqref{eq1_sec2}, it implies
	\bee\label{eq1_thm1}
		\bar x(k+1)=\bar x(k)-\frac{1}{n}\sum_{i\in\cA(k)}\sum_{t=l_i(k)+1}^{l_i(k+1)}\rho(t)\nabla f_i(\bar x(k)+\epsilon_i(k))
	\ene
	where $\cA(k)$ is the set of activated nodes at time $t(k)$, i.e., $\cA(k)=\{i|t(k)\in\cT_i, i\in\cV\}$. Then, two cases are considered.

		{\em Case 1}:~ If $\cA(k)$ is a singleton, it follows from Theorem \ref{prop2} that $\bar x(k)$ converges to some optimal solution of the DOP in \eqref{original} where Assumption \ref{assum5} holds by using Lemma \ref{lemma1}.

	{\em Case 2}:~ If $\cA(k)$ includes multiple elements, say $i_1,...,i_v$, where $v=|\cA(k)|\leq n$, then we can incrementally compute \eqref{eq1_thm1} via $v$ iterates. Specifically, let $x^{(0)}=\bar{x}(k)$ and $x^{(v)}=\bar{x}(k+1)$, we can rewrite \eqref{eq1_thm1} as
	\bee\label{eq2_thm1}
		x^{(u)}=x^{(u-1)}-\sum_{t=l_{i_u}(k)+1}^{l_{i_u}(k+1)}\frac{\rho(t)}{n}\nabla f_{i_u}(x^{(u-1)}+\epsilon_u(k))
	\ene
	for all $u=\{1,...,v\}$, where $\epsilon_u(k)=\epsilon(k)+\bar x(k)-x^{(u-1)}$. It can be readily verified that $\sum_{k=1}^{\infty}\rho(k)|\epsilon_u(k)|<\infty$, and hence Assumption \ref{assum5} is satisfied. The desired result follows from Theorem \ref{prop2} again.
\end{proof}

Next, we evaluate the convergence rate of the AsySPA under stepsizes $\rho(k)=k^{-\alpha},\alpha\in[0.5,1]$.
\begin{theo}\label{theo2}
	Suppose Assumptions \ref{assum}-\ref{assum6} hold. Let $\rho(k)=k^{-\alpha}$ in the AsySPA and $z_i(k)$ be given in \eqref{eq1_sec2}. The following statements hold.
	\begin{enumerate}[label=(\alph*)]
		\item If $\alpha\in(0.5,1]$ and the initial condition $x_i(0)=0,\forall i\in\cV$,  then
		      \bea\label{eq1_thm2}
			      &\min_{1\leq t\leq bk}f(z_i(t))-f^\star\\ &\quad\leq\frac{4b^3c}{s(k)}\left(\frac{2(c_{\epsilon}+3b^2c)}{2\alpha-1}+3b^2cd(0)^2\right), \forall i\in\cV,
		      \ena
		      where $s(k)$ is defined in \eqref{conrate} and $c_{\epsilon}$ is defined in \eqref{eq_ce}. Moreover, if $\alpha=0.5$, then
		      \bea
			      &\min_{1\leq t\leq bk}f(z_i(t))-f^\star\\
			      &\quad\leq\frac{\ln(k)}{\sqrt{k}}\left(8b^3cc_{\epsilon}+12b^5c^2(d(0)^2+2)\right), \forall i\in\cV.
		      \ena
		\item Suppose that $f_i(x),i\in\cV$ also satisfies Assumption \ref{assum8} with $\theta\in[0.5,1]$. There exists some $c_b\in(0,\infty)$ such that
		      \bee
			      d(z_i(k))^2\leq c_b k^{-2\theta\alpha},\ \forall i\in\cV
		      \ene
		      where $c_b$ depends on the parameters $\bar{\tau},\underline{\tau},\tau$ and $c$ is defined in Assumptions \ref{assum}.
	\end{enumerate}
\end{theo}
\diff{
\begin{remark}
	As in Theorem \ref{theo1}, the global index $k$ is adopted to evaluate convergence rate. This is a common and natural way to study asynchronous algorithms, see e.g. \cite{assran2018asynchronous,tian2018asy,bertsekas1989parallel}. Note that the convergence rate is evaluated under the worst case (c.f. Lemma \ref{lemma1}), and in practice it may be much faster if the unevenness of update rates and delays are relatively small. One may consider to evaluate the convergence rate in terms of the running time, the maximum number of local iterations among nodes, and etc. For synchronous distributed algorithms, a lower bound on the number of communication rounds and gradient computations has been studied in \cite{uribe2018dual}.
	
	In view of Theorem \ref{theo2}(a) and \cite{nedic2015distributed}, it is interesting that the AsySPA essentially does not lead to any rate reduction of convergence with respect to the index $k$.  In the AsySPA, $k$ is increased by one no matter on which node has completed an update. While in the SynSPA, $k$ is increased by one only if {\em all} nodes have completed their updates, as shown in Fig. \ref{fig3}. That is, the SynSPA will be delayed by the slowest node. From this perspective, the computation time for increasing $k$ in the SynSPA should be longer than that of the AsySPA. In fact, the larger  the asynchrony among computing nodes, the more significant improvement of the AsySPA over the SynSPA. This is also confirmed via simulation in Section \ref{sec6}.
\end{remark}}

{\em Proof of Theorem \ref{theo2}}: The proof of part (a) lies in the use of Theorem \ref{prop3} on \eqref{eq1_thm1} or \eqref{eq2_thm1}. We first show that there exists some $c_{\epsilon}\in(0,\infty)$ such that $|z_i(k+1)-\bar{x}(k)|\leq c_{\epsilon}\rho(k)$ for all $k$. In view of Lemma \ref{lemma3} and $x_i(0)=0$, we have
\bea\label{eq4_thm2}
	&|z_i(k+1)-\bar{x}(k)|\\
	&\leq8n^{nb}\left(\lambda^k\|\vec x(1)\|_1+\sum_{t=1}^{k}\lambda^{k-t}\|\vec g(t)\|_1\right)\\
	&\leq8n^{nb}\left(\lambda^k\|\vec x(1)\|_1+\frac{ncb^{\alpha}}{1-\lambda}k^{-\alpha}\right)\leq c_{\epsilon}\rho(k)
\ena
where
\bee\label{eq_ce}
	c_{\epsilon}=\frac{8n^{nb+1}cb^{\alpha}}{1-\lambda},
\ene
and the second inequality follows from
\bea
	\sum_{t=1}^{k}\lambda^{k-t}\|\vec g(t)\|_1&\leq\sum_{t=1}^{k}\lambda^{k-t}nc\rho(\lfloor \frac{k}{b}\rfloor)\\
	&\leq ncb^{\alpha}\sum_{t=1}^{k}\lambda^{k-t}k^{-\alpha}\leq \frac{ncb^{\alpha}}{1-\lambda}k^{-\alpha}.
\ena

By Theorem \ref{prop3} and \eqref{eq1_thm1}, we obtain
\bea\label{eq3_thm2}
	&\min_{1\leq t\leq kb}f(\bar{x}(t))-f^\star\\
	&\quad\leq\frac{4b^3c}{s(k)}\left(\frac{2(c_{\epsilon}+3b^2c)}{2\alpha-1}+3b^2cd(\vec x(1))^2\right).
\ena

Moreover, $f(z_i(t+1))-f(\bar{x}(t))\leq |\nabla f(z_i(t+1))|\cdot |z_i(t+1)-\bar{x}(t)|\le c |z_i(t+1)-\bar{x}(t)|$. Jointly with \eqref{eq4_thm2} and \eqref{eq3_thm2}, the desired result follows as $\rho(k)$ decreases faster than $1/s(k)$ if $\alpha\in(0.5,1]$.
Part (b) can be established by using Theorem \ref{prop3}(b) in a similar way. \qed

Theorem \ref{theo2}(b) provides a better convergence rate of the AsySPA if each local objective function further satisfies the H\"{o}lder error bound condition. \diff{However, the AsySPA usually cannot achieve linear rate even if $f_i(x)$ is strongly convex with Lipschize continuous gradient. One may expect that the convergence rate of the AsySPA is lower and upper bounded by two SynSPAs. The worst case corresponds to the SynSPA in which each node updates in the same rate as that of the slowest node, and the best case is the SynSPA in which each node updates in the same rate as that of the fastest node. Since the SynSPA is unable to achieve linear rate, the AsySPA cannot neither. In our recent work \cite{zhang2019asynchronous},  we show how to achieve linear rate by substantially revising the AsySPA. }

\section{Numerical Examples}\label{sec6}
In this section we numerically illustrate the performance of the AsySPA in distributedly training a multi-class logistic regression classifier on the \emph{Covertype} dataset, which is available from the \emph{UCI machine learning repository} \cite{Dua2017UCI} and has also been used in \cite{assran2018asynchronous}. \diff{Moreover, we further test the algorithms on a support vector machine (SVM) training problem in Section \ref{sec6e}}.

The objective is to solve the optimization problem
\bee\label{lr}
	\minimize_{X}\ \diff{-}\sum_{i=1}^{n_s}\sum_{j=1}^{n_c}l_j^i\log\left(\frac{\exp(\vec x_j^\T\vec s^i)}{\sum_{j'=1}^{n_c}\exp(\vec x_{j'}^\T\vec s^i)}\right)+\frac{\gamma}{2}\|X\|_F^\diff{2}
\ene
where $n_s=581012$ is the number of training instances, $n_c=7$ is the number of classes, $n_f=55$ is the number of features, $\vec s^i\in\bR^{55}$ is the feature vector of the $i$-th instance, $\vec{l}^i=[l_1^i,...,l_7^i]^\T$ is the label vector of the $i$-th instance using the one-hot encoding, $X=[\vec x_1,...,\vec x_7]\in\bR^{n_f\times n_c}$ is the weighting matrix to be optimized, $\gamma=1$ is a regularization factor, and $\|\cdot\|_F$ denotes the Frobenius norm. As in \cite{assran2018asynchronous}, we normalize non-categorical features by subtracting the mean and dividing by the standard deviation.

The problem in \eqref{lr} is strictly convex. To validate the AsySPA, we adopt the Message Passing Interface (MPI) on a server to simulate a network with multiple connected nodes. Specifically, the MPI uses $n$ cores (nodes) to form a graph $\cG$, and the communication among cores is configured to be directed, synchronous, or asynchronous. Each core is exclusively assigned a subset of the dataset with $n_s/n$ training instances, and hence only has a local objective function. We perform several experiments, and the source codes in Python are available on the Internet\footnote{https://github.com/jiaqi61/asyspa}.

We use the averaged training error, i.e., $\frac{1}{n_s}\left(f(\bar X(t))-\widehat f^\star\right)$  to  compare distributed algorithms under different scenarios, where $\bar X(t)$ is the mean of all nodes' states at time $t$ and $\widehat f^\star$ is an approximated optimal value of the problem \eqref{lr} that is obtained from a centralized gradient descent method with sufficiently many iterations.

\begin{figure}[!t]
	\centering
	\subfloat[]{\label{fig5a}\includegraphics[width=0.3\linewidth]{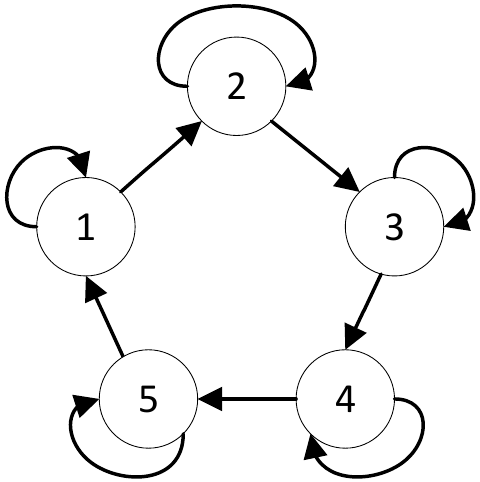}}\hspace{20pt}
	\subfloat[]{\label{fig5b}\includegraphics[width=0.3\linewidth]{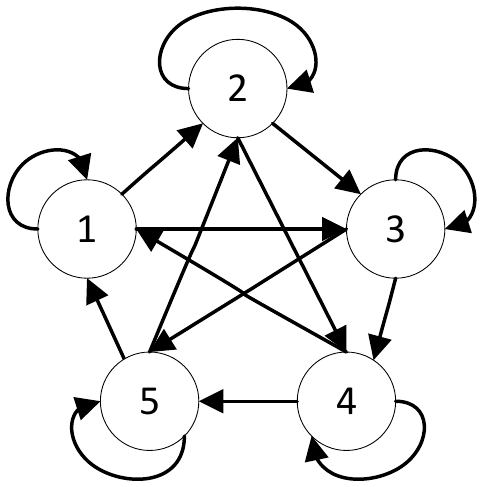}}
	\caption{Communication digraphs.}
	\label{fig5}
\end{figure}

\subsection{Comparison with centralized gradient methods}\label{sec6a}
We first compare the performance of the AsySPA with the  centralized gradient method.

The centralized algorithm is implemented on a single core while the AsySPA is tested over $n=3,6,9,12,16,20,24$ core(s) with a digraph in Fig. \ref{fig5}. Each core computes a new update by using \eqref{eq_asyspa} in Algorithm \ref{alg_asyspa} immediately after it finishes its current update. To accelerate computation, we set $\rho(k)=0.1/n_s$ in both centralized and distributed cases, rather than the diminishing stepsizes in Assumption \ref{assum}(b).

Fig. \ref{fig14} depicts how the averaged training error decreases in the computing time. Clearly, the AsySPA outperforms the centralized algorithm in all experiments and the improvement is more significant if $n$ is larger. This is further illustrated in Fig. \ref{fig15}, where we compute the speedup defined by $T_1/T_n$, and $T_n$ is the time when the averaged error over $n$ cores reduces to $10^{-2}$.

We also test the AsySPA over a digraph where a core can send messages to cores that are $2^i+1$ steps away as in \cite{lian2018asynchronous}, where $i=0,1,...,\lfloor\log_2(n-1)\rfloor$. The results remain almost unchanged from Fig. \ref{fig14}. Note that nodes in this case have a  smaller number of in-neighbors and hence it is more scalable. We further compare the effect of the number of neighbors in Section \ref{sec6d}.

\begin{figure}[!t]
	\centering
	\includegraphics[width=0.75\linewidth]{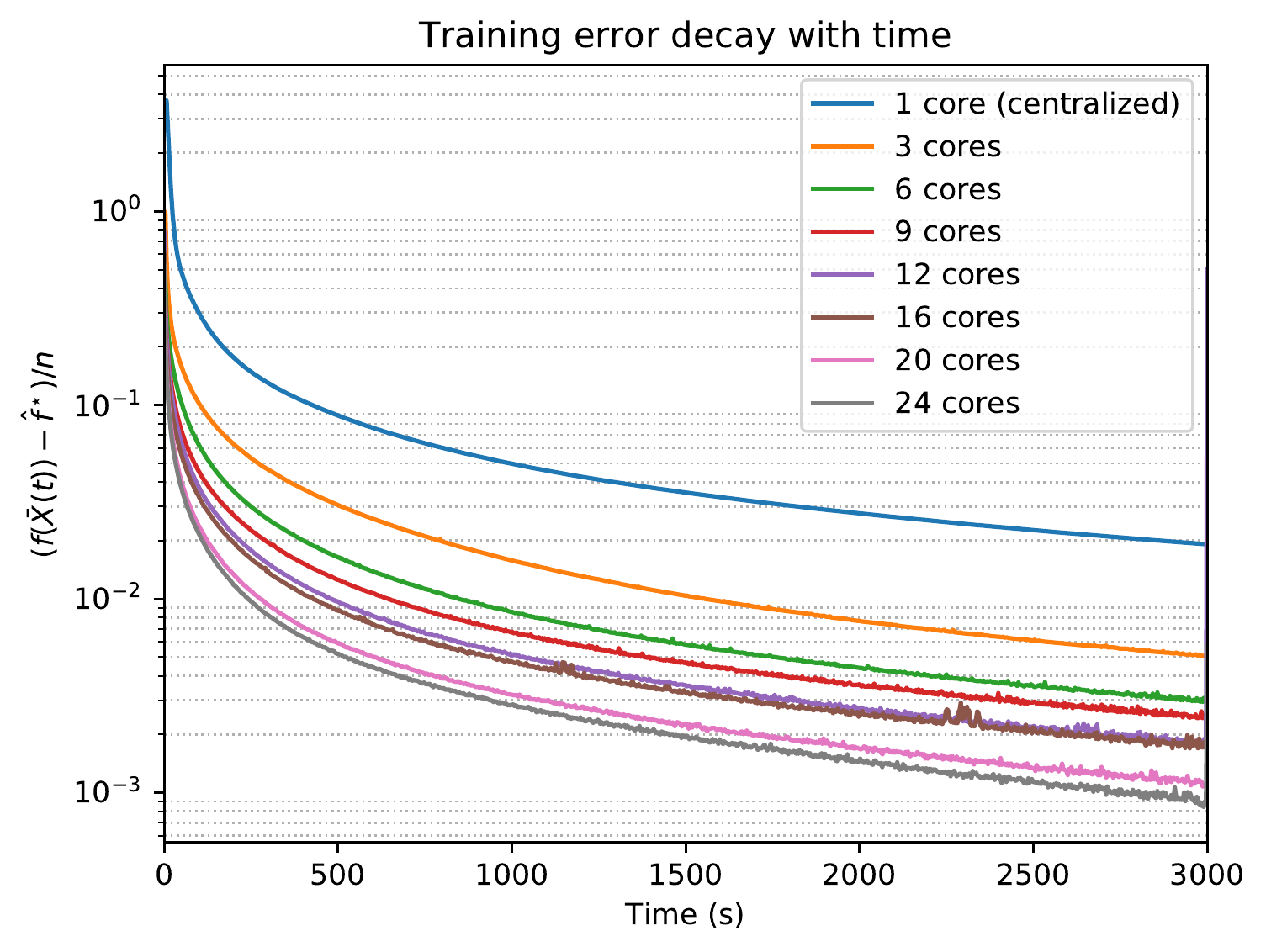}
	\caption{The averaged training errors of the AsySPA and the centralized gradient method versus the computing time.}
	\label{fig14}
\end{figure}

\begin{figure}[!t]
	\centering
	\includegraphics[width=0.75\linewidth]{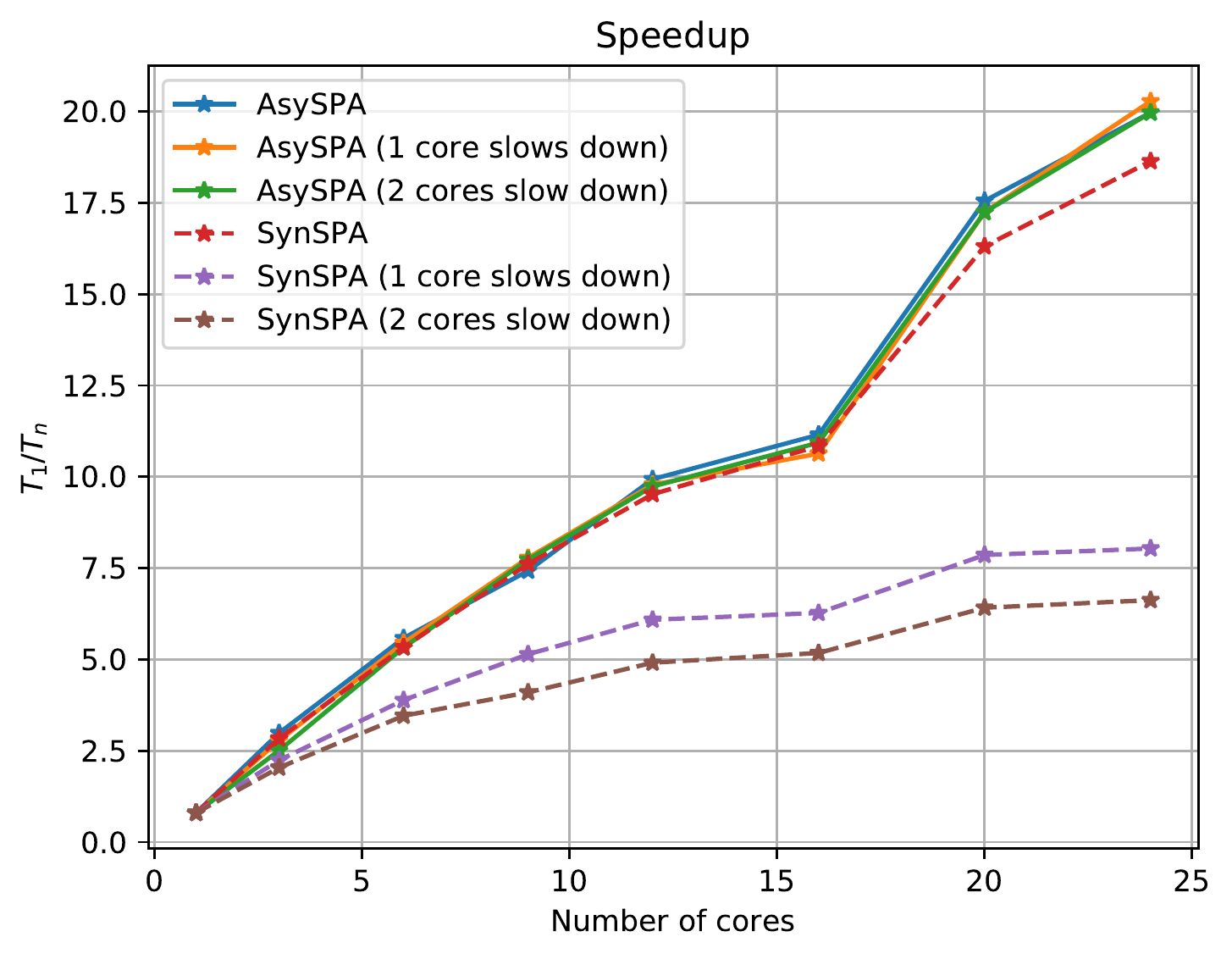}\\
	\caption{Speedup of the AsySPA and SynSPA, where $T_n$ is the time when the averaged error over $n$ cores first reduces to $10^{-2}$. Here $1$ core slows down means that we add an artificial waiting time $t_w$ after each update in one of the $n$ cores, where $t_w$ is exponentially distributed with mean $20$ms.}
	\label{fig15}
\end{figure}

\begin{figure*}[!t]
	\centering
	\subfloat[]{\label{fig7}\includegraphics[width=0.31\linewidth]{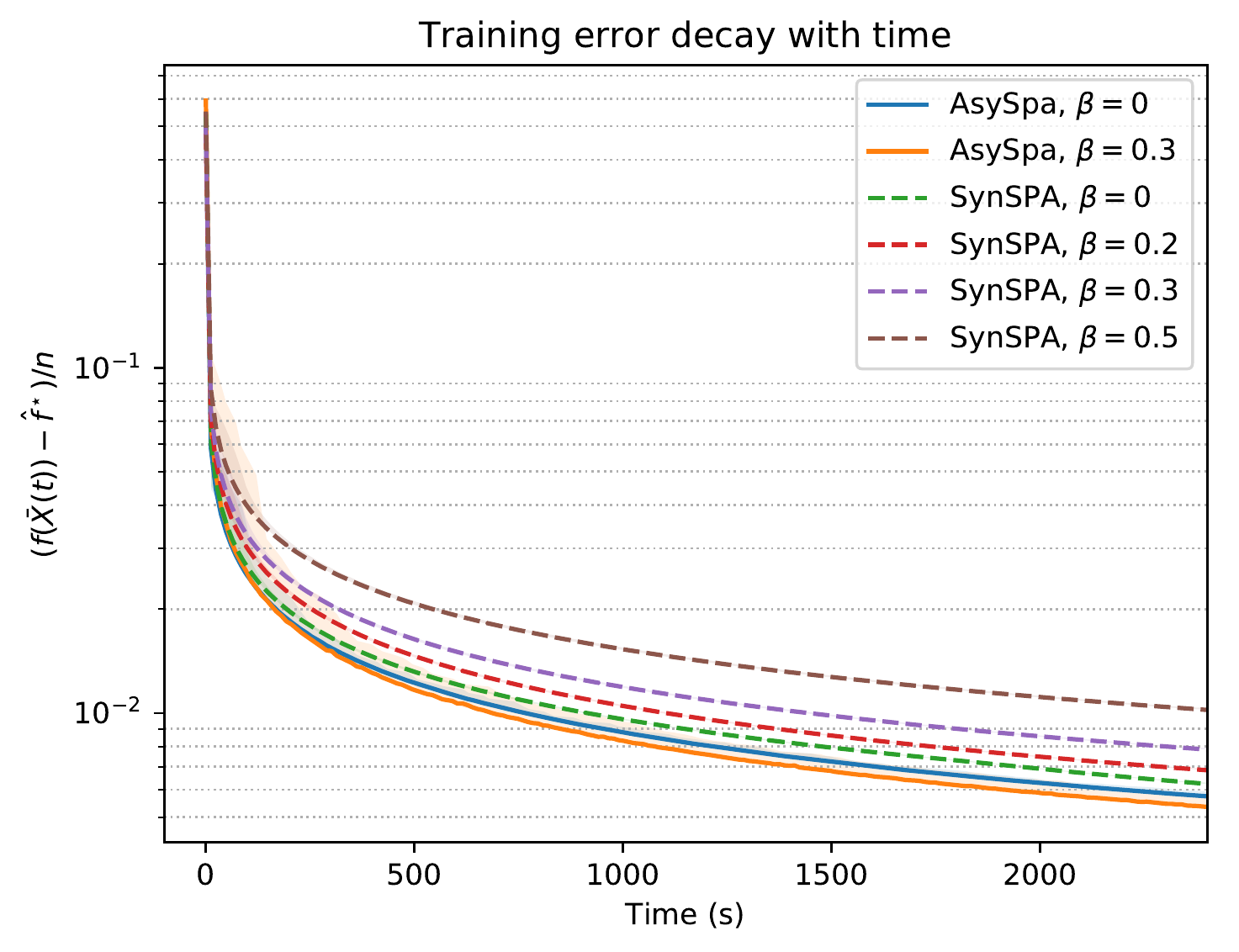}}
	\subfloat[]{\label{fig8}\includegraphics[width=0.31\linewidth]{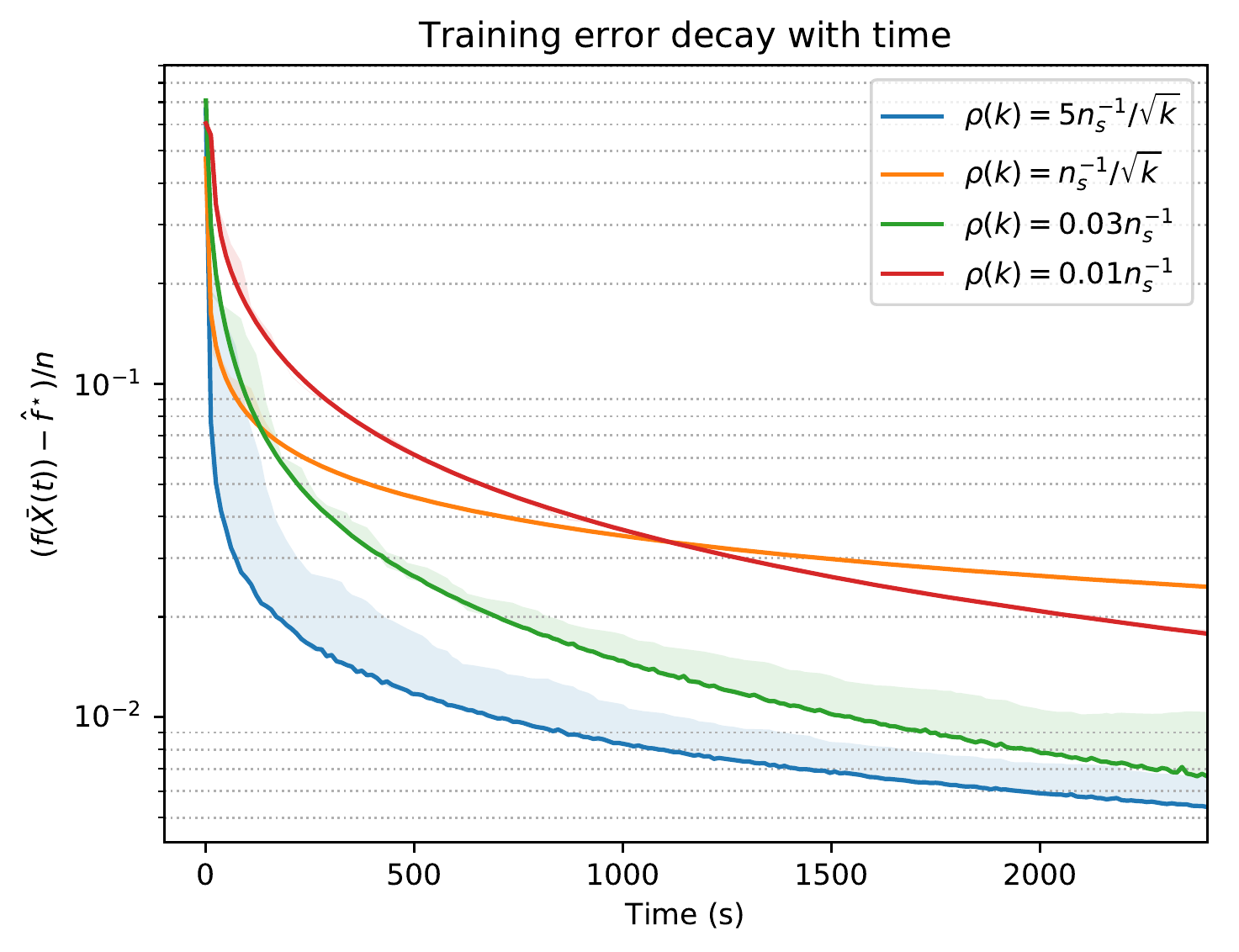}}
	\subfloat[]{\label{fig9}\includegraphics[width=0.31\linewidth]{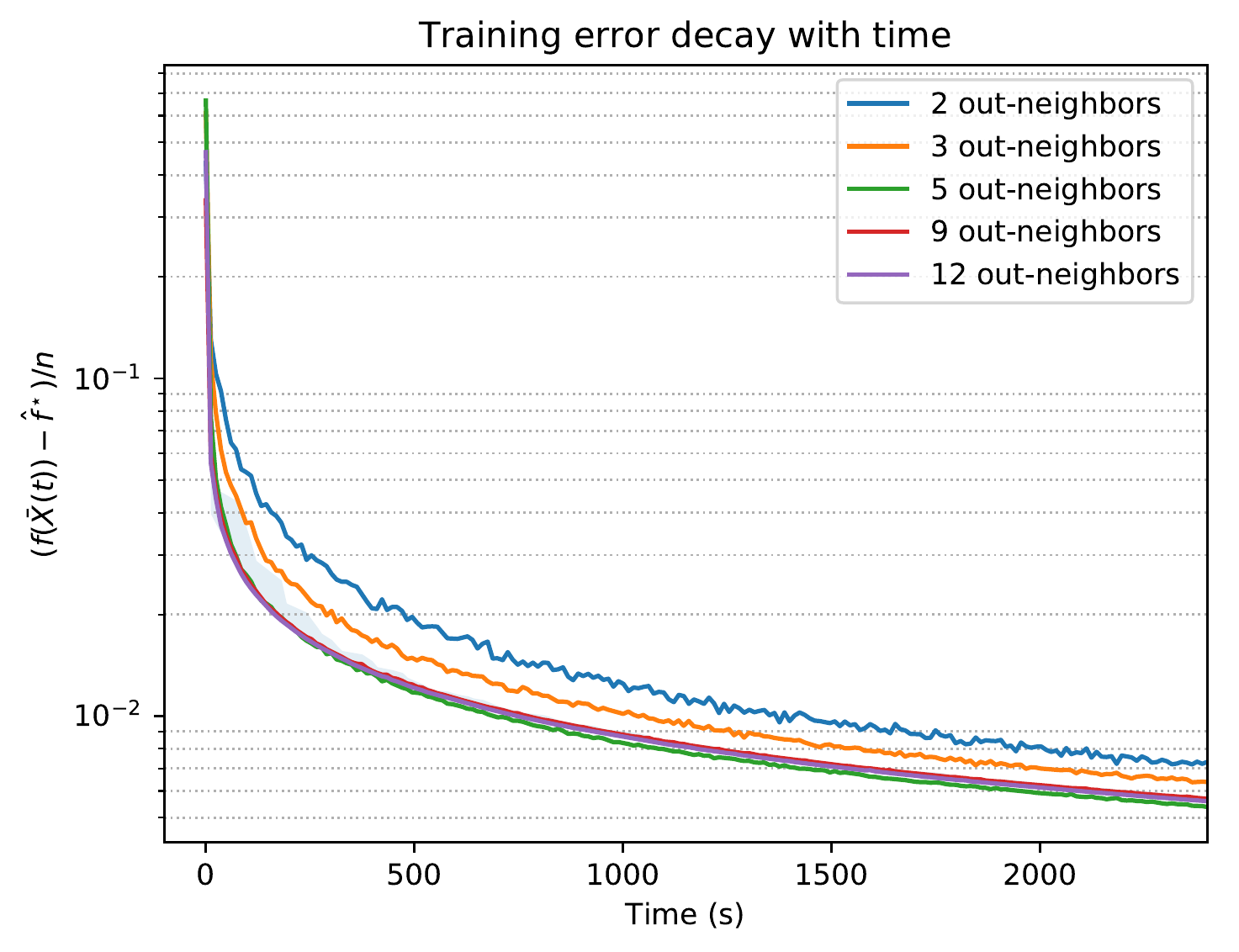}}
	\caption{(a) Convergence behaviors of the AsySPA and the SynSPA. (b) Convergence of the AsySPA under different  stepsizes. (c) Convergence of the AsySPA under different number of out-neighbors.}
	\label{fig3:exp1}
\end{figure*}

\subsection{Comparison with the SynSPA}\label{sec6b}
We also compare the AsySPA with the SynSPA in Algorithm \ref{alg_spa}, where the number of cores, communication digraph, and stepsize remain the same. To simulate the UUR, The $i$-th core in the AsySPA case updates every $i^\beta\times50$ms ($i=1,...,12$), while the $i$-th core in the synchronous case sends messages to its out-neighbors every $i^\beta\times50$ms, and perform an update after all cores have received messages from their in-neighbors. Fig. \subref*{fig7} shows the convergence behaviors of the AsySPA and the SynSPA. As expected, the AsySPA converges faster than the SynSPA in all cases.

To illustrate the better robustness of the AsySPA over the SynSPA, we consider the situation that some nodes slow down by adding an artificial waiting time $t_w$ after each update, where $t_w$ is an exponential distributed random variable with expectation 20ms. Fig. \ref{fig15} shows the speedup of the AsySPA and the SynSPA  in these experiments. Clearly, the AsySPA is almost not affected by the slowed down core(s), while the convergence rate of the SynSPA is severely reduced by the slowed down node(s).

Note that in the synchronous case the MPI acts like there is a global clock, which may be not available in practice. Moreover, the physical communication networks may have limited bandwidth and delays, which are worse than the one here inside a single computer. This makes the performance difference between the SynSPA and the AsySPA more significant.

\subsection{Comparison with the-state-of-the-art algorithm in \cite{assran2018asynchronous}}\label{sec6c}

We  compare the AsySPA with the asynchronous algorithm in \cite{assran2018asynchronous} over $12$ cores, the communication digraph of which remains the same as before. Each core uses the diminishing stepsize $\rho(k)=1/\sqrt{k}$. To simulate uneven updates of agents, we let the $i$-th core update once every $i^\beta\times 50$ms ($i=1,...,12$), and test the algorithms for $\beta = 0,0.3,0.6$. Note that a larger $\beta$ leads to a higher unevenness, and $\beta=0$ means all cores update at the same rate.

Fig. \ref{fig6} shows the averaged training error, where the shade area is the range that $f(X_i(t)),i=1,...,12$ lie within. Clearly, the algorithm in \cite{assran2018asynchronous} fails to converge to an exact optimal solution of problem \eqref{lr} except for the  case $\beta=0$, and the larger $\beta$, the larger gap from an optimal solution, which is consistent with \cite{assran2018asynchronous}. In contrast, the AsySPA converges to an optimal solution in all cases, and the unevenness only affects the maximum difference among agents' states, which however decreases to $0$.
\begin{figure}[!t]
	\centering
	\includegraphics[width=0.75\linewidth]{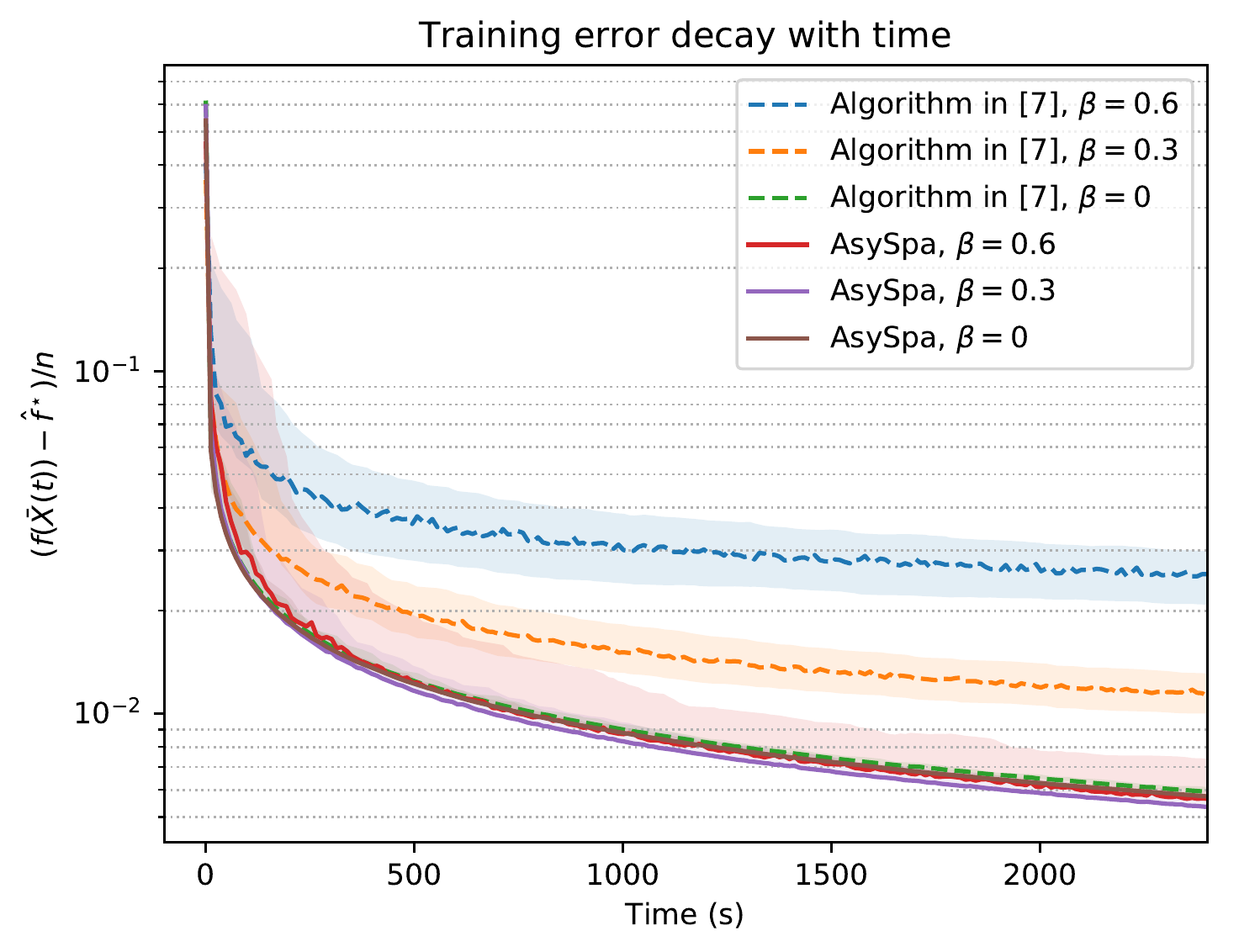}\\
	\caption{The average training errors of the AsySPA and the algorithm in \cite{assran2018asynchronous} versus computing time under different $\beta$.}
	\label{fig6}
\end{figure}

\subsection{Effect of stepsizes and number of out-neighbors}\label{sec6d}
We compare the performance of the AsySPA under different choices of stepsize. Fig. \subref*{fig8} depicts the training error under constant stepsize and diminishing stepsize. Clearly, a diminishing stepsize $\rho(k)=a/\sqrt{k}$ with larger $a$ converges faster at the expense of more oscillations (the shade area), and a constant stepsize cannot guarantee the consensus among agents. This phenomenon also appears in the standard subgradient method.

We then examine the effect of the number of out-neighbors on the performance of the AsySPA. We test the AsySPA under cases that each core sends messages to itself and its next $n_o=1,2,4,8,11$ cores (c.f. Fig. \ref{fig5}). The results are shown in Fig. \subref*{fig9}. It is interesting to notice that the algorithm converges faster by increasing the number of out-neighbors until it achieves some critical number (4 in this experiment). Then the convergence rate remains almost the same or even decreases a bit as the number of out-neighbors increases. A possible reason is that while more out-neighbors lead to less iterations to achieve consensus among agents, they increase the communication overhead.

\diff{
\subsection{Train a support vector machine}\label{sec6e}
\begin{figure}[!t]
	\centering
	\includegraphics[width=0.75\linewidth]{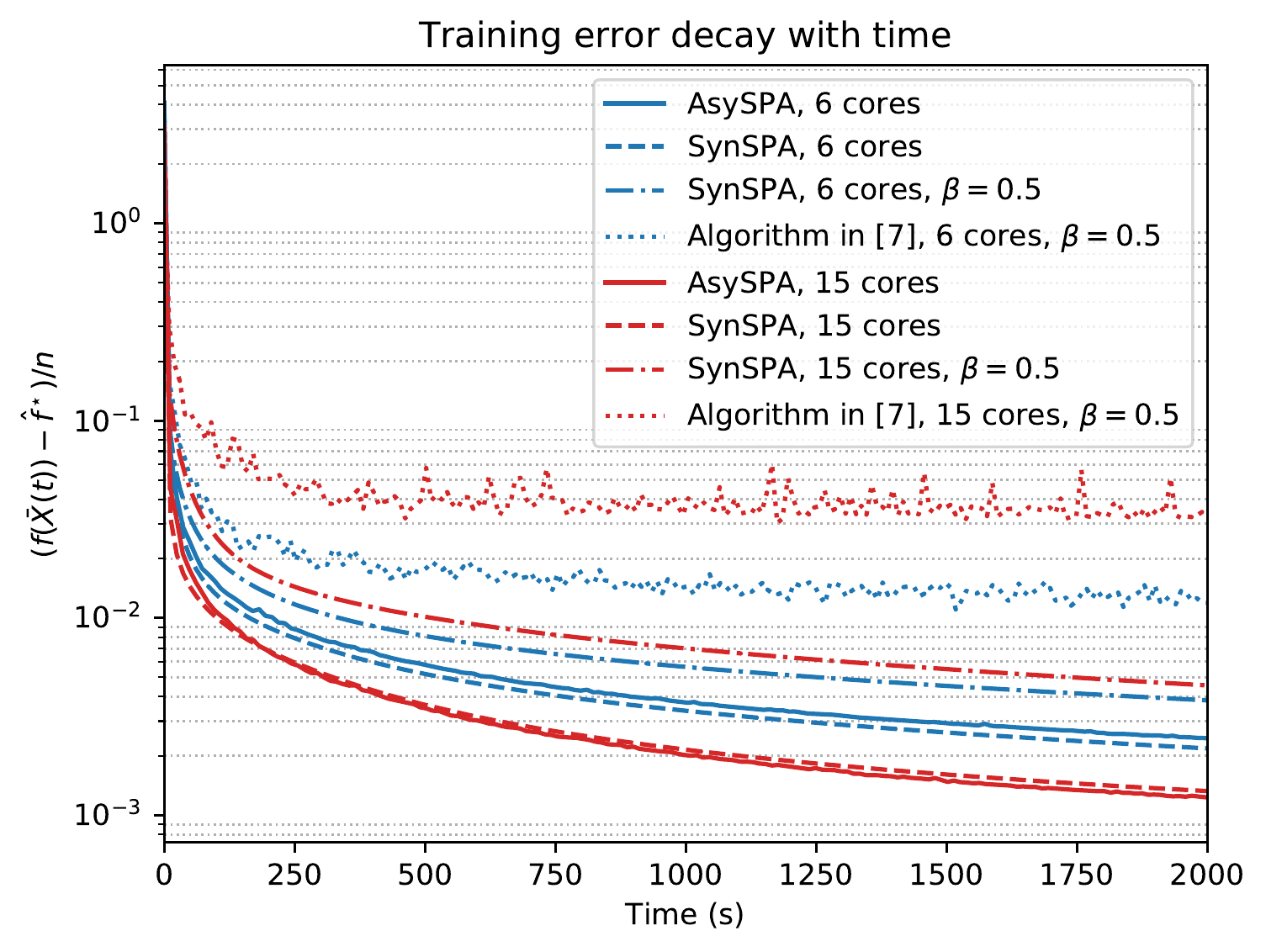}\\
	\caption{Comparison of the algorithms on a SVM training problem, where $\beta$ is defined in Section \ref{sec6c}.}
	\label{fig11}
\end{figure}

Since the subgradient methods are very suitable for minimizing non-smooth objective functions but the logistic regression cost function is smooth, we finally test the algorithm on a SVM training problem with hinge loss, which has the following non-smooth cost function,
\bee\label{svm}
	\minimize_{X}\ \sum_{i=1}^{n_s}\sum_{j\neq l^i}^{n_c}\max\left\{0,\vec x_j^\T\vec s^i-\vec x_{l^i}^\T\vec s^i+1\right\}+\frac{\gamma}{2}\|X\|_F^2
\ene
where $l^i$ is the integer label of the $i$-th instance. We repeat some of the experiments in previous subsections on this problem and the result is shown in Fig. \ref{fig11}. We can see the performance is very similar with that on the logistic regression problem.
}

\section{Conclusions}\label{sec7}

This paper has proposed the AsySPA to solve the DOP over digraphs with uneven update rates and bounded transmission delays. We also demonstrated its advantages over the existing distributed algorithms. Currently, we only consider the general case that the objective function in the DOP is convex. In the future, we shall design the accelerated AsySPA for the strongly convex objective function with Lipschitz continuous gradient.

\appendices
\section{Proof of Theorem \ref{prop2}}\label{appendix_a}
\diff{The proof can be roughly divided into 4 steps. We first show that \eqref{ism} is an $\epsilon(k)$-subgradient step in Step 1, and derive a basic equality for the distance between $\vec x(k)$ and $\vec x^\star$ in Step 2. Then, this distance is shown to be essentially decreasing in Step 3. Finally, we prove the convergence by following a similar argument in standard subgradient methods.}

	\diff{{\bf Step 1: The iterate in \eqref{ism} is an $\epsilon(k)$-subgradient step}}.\\
	 Let $\epsilon(k)=\max\{0,f_{s(k)}(\vec x(k))-f_{s(k)}(\vec x(k)+\varepsilon(k))+g_k^\T\varepsilon(k)\}$ and $g_k=\nabla f_{s(k)}(\vec{x}(k)+\varepsilon(k))$. Clearly, it follows from Assumption \ref{assum} that $\|g_k\|\leq c$ for all $k$ and
\bea
	0\leq\epsilon(k)  \leq \|\nabla f_{s(k)}\|\cdot\|\varepsilon(k)\|+\|g_k\|\cdot\|\varepsilon(k)\|\leq 2c\|\varepsilon(k)\|.
\ena
For all $\vec y\in\bR^n$, it holds that
\begin{align}\label{eq1_lem4}
	 & f_{s(k)}(\vec y)                                                                     \\
	 & \geq f_{s(k)}(\vec x(k)+\varepsilon(k))+g_k^\T(\vec y-\vec x(k)-\varepsilon(k))      \\
	 & \geq f_{s(k)}(\vec x(k))+g_k^\T(\vec y-\vec x(k))+f_{s(k)}(\vec x(k)+\varepsilon(k)) \\
	 & \quad-f_{s(k)}(\vec x(k))-g_k^\T\varepsilon(k)                                       \\
	 & \geq f_{s(k)}(\vec x(k))+g_k^\T(\vec y-\vec x(k))-\epsilon(k).
\end{align}

	\diff{{\bf Step 2: A basic equality for the distance between $\vec x(k)$ and $\vec x^\star$}}.\\
	 To simplify notations, we let $\sigma=\sigma_1+\sigma_2$, and use $\underline{r}(k)$ and $\bar{r}(k)$ to denote $r_{s(k)}(k)+1$ and $r_{s(k)}(k+1)$, respectively. It follows from Assumption \ref{assum5} that $\bar{r}(k+t-1)-\underline{r}(k)\leq t\sigma$ for all $t\in\bN$. Then, it follows from \eqref{ism} that
\begin{equation}\label{eq3_prop2}
	\begin{aligned}
		 & \|\vec{x}(k+\sigma)-\vec{x}^\star\|^2=\left\|\vec{x}(k)-\vec{x}^\star-\sum_{u=k}^{k+\sigma-1}\sum_{t=\underline{r}(u)}^{\bar{r}(u)}\rho(t)g_u\right\|^2 \\
		 & =\|\vec{x}(k)-\vec{x}^\star\|^2-2\left(\sum\nolimits_{u,t}\rho(t)g_u\right)^\T(\vec{x}(k)-\vec{x}^\star)                                                \\
		 & \quad+\left\|\sum\nolimits_{u,t}\rho(t)g_u\right\|^2                                                                                                    \\
		 & =\|\vec{x}(k)-\vec{x}^\star\|^2-2\left(\sum\nolimits_{u,t}\rho(t)g_u\right)^\T(\vec{x}(k)-\vec{x}^\star) +\widetilde{\rho}_1(k).
	\end{aligned}
\end{equation}
\diff{where we define the sum operator $\sum\nolimits_{u,t}$ as
$$
	\sum\nolimits_{u,t}:=\sum\nolimits_{u=k}^{k+\sigma-1}\sum\nolimits_{t=\underline{r}(u)}^{\bar{r}(u)}$$
to save space} and $\widetilde{\rho}_1(k)=\left\|\sum\nolimits_{u,t}\rho(t)g_u\right\|^2$. 

Moreover, we have that
\bea\label{eq2_lemma4}
	\widetilde{\rho}_1(k)&\leq c^2\left(\sigma\sum\nolimits_{t=\underline{r}(k)-\sigma}^{\bar{r}(k+\sigma-1)+\sigma}\rho(t)\right)^2\\
	&\leq c^2\sigma^3(\sigma+2)\sum\nolimits_{t=\underline{r}(k)-\sigma}^{\bar{r}(k+\sigma-1)+\sigma}\rho(t)^2
\ena
where the first inequality uses the relation $\|g_u\|\leq c$ for all $u\in\bN$,  $\underline{r}(u)\geq \underline{r}(k)-\sigma$ for all $u\geq k$ and $\bar r(u)\leq\bar r(k)+\sigma$ for all $u\leq k$ from Assumption \ref{assum5}. The second inequality follows from the Cauchy-Schwarz inequality. Hence, we obtain that
\bea\label{eq3_lemma4}
\sum_{k=1}^\infty\widetilde{\rho}_1(k)&\leq \sigma^4(\sigma+2)c^2\sum_{k=0}^\infty\sum_{t=\underline{r}(k\sigma+1)}^{\bar{r}((k+1)\sigma)}\rho(t)^2\\
&\leq 2\sigma^5c^2\sum_{t=1}^{\infty}\rho(t)^2<\infty.
\ena

\diff{{\bf Step 3: The distance between $\vec x(k)$ and $\vec x^\star$ is essentially decreasing}}. 

 \diff{Firstly}, we consider the second term of the right-hand-side of \eqref{eq3_prop2}, and obtain that
\bea\label{eq10_prop2}
	 & \left(\sum\nolimits_{u,t}\rho(t)g_u\right)^\T(\vec{x}(k)-\vec{x}^\star)                                          \\
	 & =\sum\nolimits_{u,t}\rho(t)g_u^\T(\vec x(u)-\vec x^\star+\vec{x}(k)-\vec{x}(u))                                  \\
	 & =\sum\nolimits_{u,t}\rho(t)g_u^\T(\vec x(u)-\vec x^\star) +\sum\nolimits_{u,t}\rho(t)g_u^\T(\vec x(k)-\vec x(u)) \\
	 & \geq\sum\nolimits_{u,t}\rho(t)(f_{s(u)}(\vec{x}(u))-f_{s(u)}(\vec x^\star)-\varepsilon(u))                            \\
	 & \quad-\sum\nolimits_{u,t}\rho(t)\|g_u\|\cdot\|\vec x(k)-\vec x(u)\|                                              \\
	 & =\sum\nolimits_{u,t}\rho(t)(f_{s(u)}(\vec{x}(u))-f_{s(u)}(\vec x^\star)-\epsilon(u))-\widetilde{\rho}_2(k).
\ena
where the first equality is from \eqref{eq1_lem4} and $$\widetilde{\rho}_2(k)=\sum\nolimits_{u,t}\rho(t)\|g_u\|\cdot\|\vec x(k)-\vec x(u)\|.$$
Note that $
	\|\vec x(k)-\vec x(u)\|\leq c\|\sum\nolimits_{u,t}\rho(t)\|,\ \forall u\leq k+b-1$, which implies that
	$\widetilde{\rho}_2(k)\leq c^2\left(\sum\nolimits_{u,t}\rho(t)\right)^2.$ Then, it follows from \eqref{eq2_lemma4} and \eqref{eq3_lemma4} that $\sum_{k=1}^\infty\widetilde{\rho}_2(k)<\infty$.

\diff{Secondly}, we evaluate the first term of the right-hand-side of \eqref{eq10_prop2}, leading to that
\bea\label{eq11_prop2}
	 & \sum\nolimits_{u,t}\rho(t)(f_{s(u)}(\vec{x}(u))-f_{s(u)}(\vec x^\star)-\epsilon(u))                                 \\
	 & =\sum\nolimits_{u,t}\rho(t)(f_{s(u)}(\vec{x}(k))-f_{s(u)}(\vec x^\star))-\sum\nolimits_{u,t}\rho(t)\epsilon(u)      \\
	 & \quad+\sum\nolimits_{u,t}\rho(t)(f_{s(u)}(\vec{x}(u))-f_{s(u)}(\vec{x}(k)))                                    \\
	 & \geq\sum\nolimits_{u,t}\rho(t)(f_{s(u)}(\vec{x}(k))-f_{s(u)}(\vec x^\star))-\sum\nolimits_{u,t}\rho(t)\epsilon(u)   \\
	 & \quad-\sum\nolimits_{u,t}\rho(t)\|\nabla f_{s(u)}(\vec{x}(k))\|\cdot\|\vec{x}(u)-\vec{x}(k)\|                  \\
	 & =\sum\nolimits_{u,t}\rho(t)(f_{s(u)}(\vec{x}(k))-f_{s(u)}(\vec x^\star))-\widetilde \rho_3(k)                       \\
	 & \quad-\sum\nolimits_{u,t}\rho(t)\epsilon(u)                                                                    \\
	 & =\sum\nolimits_{u,t}\rho(t)(f_{s(u)}(\vec{x}(k))-f_{s(u)}(\vec x^\star))-\widetilde \rho_3(k)-\widetilde \rho_4(k)
\ena
where $\widetilde \rho_3(k)=\sum\nolimits_{u,t}\rho(t)\|\nabla f_{s(u)}(\vec{x}(k))\|\cdot\|\vec{x}(u)-\vec{x}(k)\|$ and $\sum_{k=1}^\infty\widetilde{\rho}_3(k)\leq \infty$. In addition, $\widetilde \rho_4(k)=\sum\nolimits_{u,t}\rho(t)\epsilon(u)>0$, and
$
	\sum_{k=1}^\infty\widetilde{\rho}_4(k)<\sum_{k=1}^\infty\sum_{t=\underline{r}(k)-\sigma}^{\bar{r}(k+\sigma-1)+\sigma}\rho(t)\epsilon(k)<b^3\sum_{k=1}^\infty\rho(k)\epsilon(k)<\infty.
$

\diff{Thirdly}, we examine the first term of the right-hand-side of \eqref{eq11_prop2}, i.e., 
\begin{align}\label{eq13_prop2}
	 & \sum\nolimits_{u,t}\rho(t)(f_{s(u)}(\vec{x}(k))-f_{s(u)}(\vec x^\star)) \\ &=\sum\nolimits_{i=1}^{n}\sum\nolimits_{t=r_i(k)+1}^{r_i(k+\sigma)}\rho(t) (f_{i}(\vec{x}(k))-f_i(\vec x^\star))
\end{align}
where the equality holds by expanding the left-hand-side with reorganization, and then uses the fact that $\cV\subseteq\{s(k+1),...,s(k+\sigma)\}$ for any $k$ from Assumption \ref{assum5}.

Let $\widetilde r(k)=\max_jr_j(k)+1$. It follows from Assumption \ref{assum5} that $\widetilde r(k+1)\geq\widetilde r(k)$ and $\widetilde r(k+\sigma)\geq\widetilde r(k)+1$. Jointly with \eqref{eq13_prop2}, we obtain
\bea\label{eq12_prop2}
	 & \sum\nolimits_{u,t}\rho(t)(f_{s(u)}(\vec{x}(k))-f_{s(u)}(\vec x^\star))                                                                                                      \\
	 & =\sum\nolimits_{i=1}^{n}\big(\sum\nolimits_{t=r_i(k)+1}^{\widetilde r(k)}+\sum\nolimits_{t=\widetilde r(k)+1}^{\widetilde r(k+\sigma)}-\sum\nolimits_{\substack{t=1+r_i
				(k+\sigma)}}^{\widetilde r(k+\sigma)}\big)                                                                                                                                 \\
	 & \quad\times \rho(t) (f_{i}(\vec{x}(k))-f_i(\vec x^\star))                                                                                                                    \\
	 & =\sum\nolimits_{t=\widetilde r(k)+1}^{\widetilde r(k+\sigma)}\rho(t)(f(\vec{x}(k))-f(\vec x^\star))                                                                          \\
	 & \quad+\sum\nolimits_{i=1}^{n}\sum\nolimits_{t=r_i(k)+1}^{\widetilde r(k)}\rho(t) (f_{i}(\vec{x}(k))-f_i(\vec x^\star))                                                       \\
	 & \quad - \sum\nolimits_{i=1}^{n}\sum\nolimits_{\substack{t=r_i
			(k+\sigma)+1}}^{\widetilde r(k+\sigma)}\rho(t) \left(f_{i}(\vec{x}(k+\sigma))-f_i(\vec x^\star)\right)                                                                          \\
	 & \quad - \sum\nolimits_{i=1}^{n}\sum\nolimits_{\substack{t=r_i
			(k+\sigma)+1}}^{\widetilde r(k+\sigma)}\rho(t) \left(f_i(\vec{x}(k))-f_{i}(\vec{x}(k+\sigma))\right)                                                                       \\
	 & =\hspace{-0.2cm}\sum_{t=\widetilde r(k)+1}^{\widetilde r(k+\sigma)}\rho(t)(f(\vec{x}(k))-f(\vec x^\star))+\delta(k)-\delta(k+\sigma)-\widetilde \rho_5(k)
\ena
where we use the fact that $r_i(k)\geq \widetilde r(k)-\sigma$ for all $k$ and $i$ from Assumption \ref{assum5}, and
\bea\label{eq4_prop1}
	\delta(k)   & = \sum\nolimits_{i=1}^{n}\sum\nolimits_{t=r_i(k)+1}^{\widetilde r(k)}\rho(t) (f_{i}(\vec{x}(k))-f_i(\vec x^\star)),                                                        \\
	|\delta(k)| & \leq \sum\nolimits_{t=\widetilde r(k)-\sigma}^{\widetilde r(k)}\rho(t)n\max_i |f_i(\vec{x}(k))-f_i(\vec x^\star)|                                                          \\
	            & \leq\sum\nolimits_{t=\widetilde r(k)-\sigma}^{\widetilde r(k)}\rho(t)nc\|\vec x(k)-\vec x^\star\|                                                                     \\
	            & \leq \sum_{t=\widetilde r(k)-\sigma}^{\widetilde r(k)}\rho(t)nc\big(\|\vec x(1)-\vec x^\star\|+c\sum_{u=1}^{k-1}\sum_{t=\underline{r}(u)}^{\bar{r}(u)}\rho(t)\big) \\
	            & \leq \sum_{t=\widetilde r(k)-\sigma}^{\widetilde r(k)}\rho(t)nc\big(\|\vec x(1)-\vec x^\star\|+c\sum_{t=1}^{\widetilde r(k)}\rho(t)\big)
\ena
which shows that $\delta(k)$ is bounded. Moreover, 
\begin{align}
	\widetilde \rho_5(k)   & =\sum_{i=1}^{n}\sum_{t=r_i(k+b)+1}^{\widetilde r(k+\sigma)}\rho(t) (f_{i}(\vec{x}(k))-f_i(\vec{x}(k+\sigma))), \\
	|\widetilde \rho_5(k)| & \leq nc\sum\nolimits_{t=\widetilde r(k+\sigma)-\sigma}^{\widetilde r(k+\sigma)}\rho(t) \|\vec{x}(k+\sigma)-\vec{x}(k)\| \\
	                       & \leq nc^2\sum\nolimits_{t=\widetilde r(k+\sigma)-\sigma}^{\widetilde r(k+\sigma)}\rho(t)\sum\nolimits_{u,t}\rho(t)
\end{align}
which implies that $\sum_{k=1}^\infty |\widetilde \rho_5(k)|<\infty$, i.e., $\{\widetilde \rho_5(k)\}$ converges absolutely, and hence $\lim_{k\ra\infty}\sum_{t=1}^k \widetilde \rho_5(t)$ exists and is finite.

\diff{Finally},  combining \eqref{eq10_prop2}, \eqref{eq11_prop2} and \eqref{eq12_prop2} yields the key inequality
\begin{align}\label{eq5_prop2}
	 & \left(\sum\nolimits_{u,t}\rho(t)g_u\right)^\T(\vec{x}(k)-\vec{x}^\star)                                                              \\
	 & \geq\sum\nolimits_{t=\widetilde r(k)+1}^{\widetilde r(k+\sigma)}\rho(t)(f(\vec{x}(k))-f(\vec x^\star))-\sum_{i=2}^{5}\widetilde \rho_i(k) \\
	 & \quad +\delta(k)-\delta(k+\sigma)
\end{align}
where $\sum\nolimits_{t=\widetilde r(k)+1}^{\widetilde r(k+\sigma)}\rho(t)(f(\vec{x}(k))-f(\vec x^\star))$ is positive.  In view of \eqref{eq3_prop2}, the distance between $\vec x(k)$ and $\vec x^\star$ is essentially decreasing in the sense of neglecting all trivial terms. 

\diff{{\bf Step 4: Convergence of the iterate in \eqref{ism}}}.\\
 Eq. \eqref{eq5_prop2} together with \eqref{eq3_prop2} implies that
\begin{align}\label{eq6_prop2}
	 & \|\vec{x}(k+\sigma)-\vec{x}^\star\|^2 \leq\|\vec{x}(k)-\vec{x}^\star\|^2+ \widetilde \rho_6(k)                             \\
	 & -2\sum_{t=\widetilde r(k)+1}^{\widetilde r(k+\sigma)}\rho(t)(f(\vec x(k))-f(\vec x^\star))  -2\delta(k)+2\delta(k+\sigma).
\end{align}
where $\widetilde \rho_6(k)=\widetilde \rho_1(k)+2\sum_{i=2}^{5}\widetilde \rho_i(k)$. Apparently, $\lim_{k\ra\infty}\sum_{t=1}^k \widetilde \rho_6(t)$ exists and is finite.

Let $\vec y(k)=\vec x(k\sigma+1)$, it follows from \eqref{eq6_prop2} that
\bea\label{eq2_prop2}
	&\|\vec y(k+1)-\vec x^\star\|^2\leq\|\vec y(k)-\vec x^\star\|^2\\
	&-\sum\nolimits_{t=\widetilde r(k\sigma+1)+1}^{\widetilde r((k+1)\sigma+1)}2\rho(t)(f(\vec y(k))-f(\vec x^\star))+\eta(k)
\ena
where $
	\eta(k)=\widetilde \rho_6(kb+1)-2\delta(k\sigma+1)+2\delta((k+1)\sigma+1).
$
Since $\{\delta(k)\}$ is bounded, we have for all $k\in\bN$ that
\bea\label{eq3_prop1}
	 & \sum\nolimits_{t=0}^{k}\eta(t)                                                                             \\
	 & = \sum_{t=0}^{k}\left(2\delta(t\sigma+1) -2\delta((t+1)\sigma+1) + \widetilde \rho_6((t+1)\sigma+1)\right) \\
	 & =2\delta(1)-2\delta((k+1)\sigma+1)+\sum\nolimits_{t=1}^{k}\widetilde \rho_6(t\sigma)                       \\
	 & \leq2\delta(1)-2\delta((k+1)\sigma+1)+\sum\nolimits_{t=1}^{k}|\widetilde \rho_6(t\sigma)|<\infty.
\ena

\diff{Note that $\eta(k)$ may be negative for some $k$, and thus the convergence of $\{\vec y(k)\}$ and $\{\vec x(k)\}$ cannot be obtained directly from the supermartingale convergence theorem (SCT, \cite[Proposition A.4.4]{bertsekas2015convex}). Nonetheless, we can readily prove that $\{\|\vec y(k)-\vec x^\star\|\}$ is convergent and the convergence of $\{\vec y(k)\}$ to $\vec x^\star$ by using similar arguments of SCT. Together with the fact that $\lim_{k\rightarrow\infty}\rho(k)=0$, it implies the convergence of $\{\vec x(k)\}$ to $\vec x^\star$.}\qed

\section{Proof of Theorem \ref{prop3}}\label{appendix_b}
\subsection{Proof of Part (a):}

For convenience, we only prove the case for $\alpha\in[0.5,1)$, while the proof extends to $\alpha=1$ in a straightforward manner.

The relation \eqref{eq2_prop2} implies that
\bea\label{eq1_prop3}
	&\sum\nolimits_{t=\widetilde r(k\sigma+1)+1}^{\widetilde r((k+1)\sigma+1)}2\rho(t)(f(\vec y(k))-f(\vec x^\star))\\
	&\quad\leq\|\vec y(k)-\vec x^\star\|^2-\|\vec y(k+1)-\vec x^\star\|^2+\eta(k)
\ena
Since $\widetilde r((k+1)\sigma+1)\geq\widetilde r(k\sigma+1)+1$, and hence $\sum_{t=\widetilde r(k\sigma+1)+1}^{\widetilde r((k+1)\sigma+1)}\rho(t)\geq\rho(\widetilde r(k\sigma+1)+1)\geq\rho(k+1)$, we have
\bea\label{eq2_prop3}
	&2\rho(k+1)(f(\vec y(k))-f(\vec x^\star))\\
	&\quad\leq\|\vec y(k)-\vec x^\star\|^2-\|\vec y(k+1)-\vec x^\star\|^2+\eta(k)
\ena
Adding \eqref{eq2_prop3} over $k$ implies that
\bee\label{eq3_prop3}
	2\sum_{t=1}^{k}\rho(t)(f(\vec y(k))-f(\vec x^\star))\leq\|\vec y(1)-\vec x^\star\|^2+\sum_{t=1}^{k}\eta(t).
\ene
It follows from \eqref{eq3_prop1}, the definitions of $\widetilde \rho_i(k),i=1,...,6$, and \eqref{eq4_prop1} that
\bea\label{eq4_prop2}
	&\sum\nolimits_{t=0}^{k}\eta(t)=2\delta(1)-2\delta((k+1)\sigma+1)+\sum\nolimits_{t=1}^{k}|\widetilde \rho_6(t\sigma)|\\
	&\leq 2\sigma^3\sum\nolimits_{t=1}^k\rho(k)\epsilon(k)+10\sigma^5c^2\sum\nolimits_{t=1}^{k}\rho(t)^2\\
	&\quad+2nc(\sigma\rho(k)d(\vec x(1))+\sigma^2\sum\nolimits_{t=1}^{k}\rho(t)^2)\\
	&\leq 2\sigma^3\sum_{t=1}^k\rho(k)\epsilon(k)+12\sigma^5c^2(\sum_{t=1}^{k}\rho(t)^2+\rho(k)d(\vec x(1)))\\
	&\leq \frac{8\sigma^3cc_{\varepsilon}\alpha}{2\alpha-1}+12\sigma^5c^2(\frac{2\alpha}{2\alpha-1}+k^{-\alpha}d(\vec x(1))^2)\\
	&\leq \frac{8\sigma^3cc_{\varepsilon}}{2\alpha-1}+12\sigma^5c^2(\frac{2}{2\alpha-1}+d(\vec x(1))^2)\\
\ena
where the second to last inequality uses $\rho(k)=k^{-\alpha},\epsilon(k)\leq 2cc_{\varepsilon}\rho(k)$ and the relation $
	\int_1^{k}\frac{1}{x^\alpha}dx<\sum_{t=1}^{k} \frac{1}{t^\alpha}<\int_1^{k}\frac{1}{x^\alpha}dx+1$,
and the last inequality follows from that $\alpha\leq1$ and $k^{-\alpha}\leq 1$.

Eq. \eqref{eq4_prop2} combined with \eqref{eq3_prop3} and the fact that $\sum_{t=1}^{k}\rho(t)>\int_1^{k}\frac{1}{x^{\alpha}}dx=s(k)$ yields
\bea
	&\min_{1\leq t\leq k}f(\vec y(t))-f(\vec x^\star)\leq\frac{\|\vec y(1)-\vec x^\star\|^2+\sum_{t=1}^{k}\eta(t)}{s(k)}\\
	&\leq\frac{4\sigma^3c}{s(k)}\left(\frac{2(c_{\varepsilon}+3\sigma^2c)}{2\alpha-1}+3\sigma^2cd(\vec x(1))^2\right).
\ena
The desired result follows by noticing that $\vec y(k)=\vec x(k\sigma+1)$.

\subsection{Proof of Part (b):}

Define $z(k)=d(\vec y(k))^2-2\delta(k\sigma+1)$ and we can rewrite \eqref{eq2_prop2} as
\bea\label{eq4_prop3}
	&z(k+1)\\
	&= \|\vec y(k+1)-\vec x^\star_{k+1}\|^2-2\delta(k\sigma+1)\\
	&\leq\|\vec y(k+1)-\vec x^\star_k\|^2-2\delta(k\sigma+1)\\
	&\leq z(k)-\sum_{t=\widetilde r(k\sigma+1)+1}^{\widetilde r((k+1)\sigma+1)}2\rho(t)(f(\vec y(k))-f^\star)+\widetilde \rho_6(k\sigma+1)\\
	&\leq z(k)-c_h\sum_{t=\widetilde r(k\sigma+1)+1}^{\widetilde r((k+1)\sigma+1)}2\rho(t)d(\vec y^k)^{1/\theta}+\widetilde \rho_6(k\sigma+1)\\
	&\leq z(k)-2c_h\rho(k\sigma)d(\vec y(k))^{1/\theta}+\widetilde \rho_6(k\sigma+1)
\ena
where $\vec x_k^\star=\argmin_{\vec x^\star\in\cX^\star}\|\vec y(k)-\vec x^\star\|$ and hence the first inequality follows; the second inequality is from \eqref{eq2_prop2}, the third inequality uses the H\"{o}lder error bound, and the last equality follows from $\widetilde r(k\sigma+1)+1\geq k\sigma$ and thus $\sum_{t=\widetilde r(k\sigma+1)+1}^{\widetilde r((k+1)\sigma+1)}\rho(t)\geq\rho(k\sigma)$.

The goal is to derive convergence rates for a sequence $\{z(k)\}$ satisfying \eqref{eq4_prop3}. Note that it follows from \eqref{eq4_prop1} that
\bea\label{eq6_prop3}
	&d(\vec y(k))^2-2\sigma^2c\rho(k)d(\vec y(k))\leq z(k)\\
	&\leq d(\vec y(k))^2+2\sigma^2c\rho(k)d(\vec y(k)).
\ena
Let
\bee
	\cK_I=\{k\in\bN|d(\vec y^k)\leq 2\sigma^2c\rho(k)\}
\ene
and its complement
\bee
	\cK_I^c=\{k\in\bN|d(\vec y^k)> 2\sigma^2c\rho(k)\}.
\ene

It holds that for any $k\in\cK_I$,
\bee
	d(\vec y(k))^2\leq 4\sigma^4c^2\rho(k)^2\leq 4\sigma^4c^2k^{-2\alpha}\leq 4\sigma^4c^2k^{-2\theta\alpha}.
\ene

\diff{If $\cK_I^c$ is empty, then the proof completes. Otherwise, let $k\in\cK_I^c$ and $k_0,k_1$ be two consecutive elements in $\cK_I$ ($k_0\leq k_1$) such that $k_0\leq k\leq k_1$}. We have from \eqref{eq6_prop3} that
\bee
	d(\vec y(k))^2\geq \frac{1}{2}d(\vec y(k))^2+\sigma^2c\rho(k)d(\vec y(k)))\geq \frac{1}{2}\vec z(k)\geq 0.
\ene
Then, it follows from \eqref{eq4_prop3} that
\bea\label{eq5_prop3}
	z(k+1)&\leq z(k)-c_h\rho(k\sigma)z(k)^{1/2\theta}+\widetilde \rho_6(k\sigma+1)\\
	&\leq z(k)-c_h\rho(k\sigma)z(k)^{1/2\theta}+\sigma^4c^2 \rho(k\sigma)^2
\ena
Invoking Theorem 6 in \cite{johnstone2017faster}, we obtain that there exists a constant $c_b^{(1)}$ such that
$
	z(k)\leq c_b^{(1)}k^{-2\theta\alpha}.
$

On the other hand, \eqref{eq4_prop1} implies that
\begin{align}\label{eq8_prop2}
	 & \delta(k\sigma+1)                                                                                                               \\
	 & = \sum_{i=1}^{n}\sum_{t=r_i(k\sigma+1)+1}^{\widetilde r(k\sigma+1)}\rho(k\sigma+1) (f_{i}(\vec{x}(k\sigma+1))-f_i(\vec x^\star))     \\
	 & \leq  \sum_{i=1}^{n}\sum_{t=r_i(k\sigma+1)+1}^{\widetilde r(k\sigma+1)}\rho(t)ncd(\vec y^k)\leq c_b^{(2)}d(\vec y^k)k^{-\alpha}
\end{align}
where $c_b^{(2)}$ is a positive constant. Thus, we have
\bee
	d(\vec y(k))^2=z(k)+2\delta(k\sigma+1)\leq c_b^{(1)}k^{-2\theta\alpha}+c_b^{(2)}d(\vec y^k)k^{-\alpha}
\ene
which implies $d(\vec y(k))^2\leq c_b k^{-2\theta\alpha}$ for some $c_b>0$.\qed

\bibliographystyle{IEEEtran}
\bibliography{mybibf}         

\begin{thebibliography}{10}
\providecommand{\url}[1]{#1}
\csname url@samestyle\endcsname
\providecommand{\newblock}{\relax}
\providecommand{\bibinfo}[2]{#2}
\providecommand{\BIBentrySTDinterwordspacing}{\spaceskip=0pt\relax}
\providecommand{\BIBentryALTinterwordstretchfactor}{4}
\providecommand{\BIBentryALTinterwordspacing}{\spaceskip=\fontdimen2\font plus
\BIBentryALTinterwordstretchfactor\fontdimen3\font minus
  \fontdimen4\font\relax}
\providecommand{\BIBforeignlanguage}[2]{{%
\expandafter\ifx\csname l@#1\endcsname\relax
\typeout{** WARNING: IEEEtran.bst: No hyphenation pattern has been}%
\typeout{** loaded for the language `#1'. Using the pattern for}%
\typeout{** the default language instead.}%
\else
\language=\csname l@#1\endcsname
\fi
#2}}
\providecommand{\BIBdecl}{\relax}
\BIBdecl

\bibitem{nedic2017network}
A.~Nedi{\'c}, A.~Olshevsky, and M.~G. Rabbat, ``Network topology and
  communication-computation tradeoffs in decentralized optimization,''
  \emph{Proceedings of the IEEE}, vol. 106, no.~5, pp. 953--976, May 2018.

\bibitem{hannah2018abcd}
\BIBentryALTinterwordspacing
R.~Hannah, F.~Feng, and W.~Yin, ``A2{BCD}: Asynchronous acceleration with
  optimal complexity,'' in \emph{International Conference on Learning
  Representations}, 2019. [Online]. Available:
  \url{https://openreview.net/forum?id=rylIAsCqYm}
\BIBentrySTDinterwordspacing

\bibitem{xie2018fast}
K.~Xie, Q.~Cai, and M.~Fu, ``A fast clock synchronization algorithm for
  wireless sensor networks,'' \emph{Automatica}, vol.~92, pp. 133--142, 2018.

\bibitem{assran2017empirical}
M.~Assran and M.~Rabbat, ``An empirical comparison of multi-agent optimization
  algorithms,'' in \emph{IEEE Global Conference on Signal and Information
  Processing}, 2017, pp. 573--577.

\bibitem{bertsekas1989parallel}
D.~P. Bertsekas and J.~N. Tsitsiklis, \emph{Parallel and distributed
  computation: numerical methods}.\hskip 1em plus 0.5em minus 0.4em\relax
  Athena Scientific, 2015.

\bibitem{nedic2015distributed}
A.~Nedi{\'c} and A.~Olshevsky, ``Distributed optimization over time-varying
  directed graphs,'' \emph{IEEE Transactions on Automatic Control}, vol.~60,
  no.~3, pp. 601--615, 2015.

\bibitem{assran2018asynchronous}
M.~Assran and M.~Rabbat, ``Asynchronous subgradient-push,'' \emph{arXiv
  preprint arXiv:1803.08950}, 2018.

\bibitem{nedic2009distributed}
A.~Nedi{\'c} and A.~Ozdaglar, ``Distributed subgradient methods for multi-agent
  optimization,'' \emph{IEEE Transactions on Automatic Control}, vol.~54,
  no.~1, pp. 48--61, 2009.

\bibitem{shi2015extra}
W.~Shi, Q.~Ling, G.~Wu, and W.~Yin, ``Extra: An exact first-order algorithm for
  decentralized consensus optimization,'' \emph{SIAM Journal on Optimization},
  vol.~25, no.~2, pp. 944--966, 2015.

\bibitem{qu2017harnessing}
G.~Qu and N.~Li, ``Harnessing smoothness to accelerate distributed
  optimization,'' \emph{IEEE Transactions on Control of Network Systems},
  vol.~5, no.~3, pp. 1245--1260, 2018.

\bibitem{zhang2017distributed}
J.~Zhang, K.~You, and T.~Ba{\c{s}}ar, ``Distributed discrete-time optimization
  in multi-agent networks using only sign of relative state,'' \emph{IEEE
  Transactions on Automatic Control}, vol.~64, no.~6, pp. 2352--2367, 2019.

\bibitem{magnusson2017convergence}
S.~Magn{\'u}sson, C.~Enyioha, N.~Li, C.~Fischione, and V.~Tarokh, ``Convergence
  of limited communications gradient methods,'' \emph{IEEE Transactions on
  Automatic Control}, vol.~63, no.~5, pp. 1356--1371, 2018.

\bibitem{xie2018distributed}
P.~Xie, K.~You, R.~Tempo, S.~Song, and C.~Wu, ``Distributed convex optimization
  with inequality constraints over time-varying unbalanced digraphs,''
  \emph{IEEE Transactions on Automatic Control}, vol.~63, no.~12, pp.
  4331--4337, 2018.

\bibitem{xi2017dextra}
C.~Xi and U.~A. Khan, ``Dextra: A fast algorithm for optimization over directed
  graphs,'' \emph{IEEE Transactions on Automatic Control}, vol.~62, no.~10, pp.
  4980--4993, 2017.

\bibitem{nedic2017achieving}
A.~Nedi{\'c}, A.~Olshevsky, and W.~Shi, ``Achieving geometric convergence for
  distributed optimization over time-varying graphs,'' \emph{SIAM Journal on
  Optimization}, vol.~27, no.~4, pp. 2597--2633, 2017.

\bibitem{notarnicola2017asynchronous}
I.~Notarnicola and G.~Notarstefano, ``Asynchronous distributed optimization via
  randomized dual proximal gradient,'' \emph{IEEE Transactions on Automatic
  Control}, vol.~62, no.~5, pp. 2095--2106, 2017.

\bibitem{nedic2011asynchronous}
A.~Nedi{\'c}, ``Asynchronous broadcast-based convex optimization over a
  network,'' \emph{IEEE Transactions on Automatic Control}, vol.~56, no.~6, pp.
  1337--1351, 2011.

\bibitem{bianchi2016coordinate}
P.~Bianchi, W.~Hachem, and F.~Iutzeler, ``A coordinate descent primal-dual
  algorithm and application to distributed asynchronous optimization,''
  \emph{IEEE Transactions on Automatic Control}, vol.~61, no.~10, pp.
  2947--2957, 2016.

\bibitem{farina2018asynchronous}
F.~Farina, A.~Garulli, A.~Giannitrapani, and G.~Notarstefano, ``Asynchronous
  distributed method of multipliers for constrained nonconvex optimization,''
  in \emph{2018 European Control Conference (ECC)}.\hskip 1em plus 0.5em minus
  0.4em\relax IEEE, 2018, pp. 2535--2540.

\bibitem{wei20131}
E.~Wei and A.~Ozdaglar, ``On the $o(1/k)$ convergence of asynchronous
  distributed alternating direction method of multipliers,'' in \emph{IEEE
  Global Conference on Signal and Information Processing}, 2013, pp. 551--554.

\bibitem{xu2018convergence}
J.~Xu, S.~Zhu, Y.~C. Soh, and L.~Xie, ``Convergence of asynchronous distributed
  gradient methods over stochastic networks,'' \emph{IEEE Transactions on
  Automatic Control}, vol.~63, no.~2, pp. 434--448, 2018.

\bibitem{peng2016arock}
Z.~Peng, Y.~Xu, M.~Yan, and W.~Yin, ``Arock: an algorithmic framework for
  asynchronous parallel coordinate updates,'' \emph{SIAM Journal on Scientific
  Computing}, vol.~38, no.~5, pp. A2851--A2879, 2016.

\bibitem{wu2018decentralized}
T.~Wu, K.~Yuan, Q.~Ling, W.~Yin, and A.~H. Sayed, ``Decentralized consensus
  optimization with asynchrony and delays,'' \emph{IEEE Transactions on Signal
  and Information Processing over Networks}, vol.~4, no.~2, pp. 293--307, 2018.

\bibitem{lian2018asynchronous}
X.~Lian, W.~Zhang, C.~Zhang, and J.~Liu, ``Asynchronous decentralized parallel
  stochastic gradient descent,'' in \emph{Proceedings of the 35th Conference on
  Machine Learning}, 2018, pp. 3049--3058.

\bibitem{zhao2015asynchronous}
X.~Zhao and A.~H. Sayed, ``Asynchronous adaptation and learning over
  networks-{P}art {I}: {M}odeling and stability analysis,'' \emph{IEEE
  Transactions on Signal Processing}, vol.~63, no.~4, pp. 811--826, 2015.

\bibitem{bof2017newton}
N.~{Bof}, R.~{Carli}, G.~{Notarstefano}, L.~{Schenato}, and D.~{Varagnolo},
  ``Multiagent newton–raphson optimization over lossy networks,'' \emph{IEEE
  Transactions on Automatic Control}, vol.~64, no.~7, pp. 2983--2990, July
  2019.

\bibitem{tian2018asy}
Y.~Tian, Y.~Sun, and G.~Scutari, ``{ASY-SONATA}: Achieving linear convergence
  in distributed asynchronous multiagent optimization,'' in \emph{2018 56th
  Annual Allerton Conference on Communication, Control, and Computing
  (Allerton)}.\hskip 1em plus 0.5em minus 0.4em\relax IEEE, 2018, pp. 543--551.

\bibitem{tsitsiklis1986distributed}
J.~Tsitsiklis, D.~Bertsekas, and M.~Athans, ``Distributed asynchronous
  deterministic and stochastic gradient optimization algorithms,'' \emph{IEEE
  Transactions on Automatic Control}, vol.~31, no.~9, pp. 803--812, 1986.

\bibitem{cannelli2017asynchronous}
L.~Cannelli, F.~Facchinei, V.~Kungurtsev, and G.~Scutari, ``Asynchronous
  parallel algorithms for nonconvex optimization,'' \emph{Mathematical
  Programming}, 2019.

\bibitem{eisen2017decentralized}
M.~Eisen, A.~Mokhtari, and A.~Ribeiro, ``Decentralized quasi-newton methods,''
  \emph{IEEE Transactions on Signal Processing}, vol.~65, no.~10, pp.
  2613--2628, 2017.

\bibitem{li1987asymptotic}
S.~Li and T.~Ba{\c{s}}ar, ``Asymptotic agreement and convergence of
  asynchronous stochastic algorithms,'' \emph{IEEE Transactions on Automatic
  Control}, vol.~32, no.~7, pp. 612--618, 1987.

\bibitem{doan2017impact}
T.~T. Doan, C.~L. Beck, and R.~Srikant, ``On the convergence rate of
  distributed gradient methods for finite-sum optimization under communication
  delays,'' \emph{Proceedings of the ACM on Measurement and Analysis of
  Computing Systems}, vol.~1, no.~2, pp. 37:1--37:27, 2017.

\bibitem{wang2015cooperative}
H.~Wang, X.~Liao, T.~Huang, and C.~Li, ``Cooperative distributed optimization
  in multiagent networks with delays,'' \emph{IEEE Transactions on Systems,
  Man, and Cybernetics: Systems}, vol.~45, no.~2, pp. 363--369, 2015.

\bibitem{yang2017distributed}
T.~Yang, J.~Lu, D.~Wu, J.~Wu, G.~Shi, Z.~Meng, and K.~H. Johansson, ``A
  distributed algorithm for economic dispatch over time-varying directed
  networks with delays,'' \emph{IEEE Transactions on Industrial Electronics},
  vol.~64, no.~6, pp. 5095--5106, 2017.

\bibitem{lin2016distributed}
P.~Lin, W.~Ren, and Y.~Song, ``Distributed multi-agent optimization subject to
  nonidentical constraints and communication delays,'' \emph{Automatica},
  vol.~65, pp. 120--131, 2016.

\bibitem{xie2018survey}
P.~Xie, K.~You, Y.~Hong, and L.~Xie, ``A survey of distributed convex
  optimization algorithms over networks,'' \emph{Control Theory \&
  Applications}, vol.~35, no.~7, pp. 918--927, 2018.

\bibitem{sun2016distributed}
Y.~Sun, G.~Scutari, and D.~Palomar, ``Distributed nonconvex multiagent
  optimization over time-varying networks,'' in \emph{50th Asilomar Conference
  on Signals, Systems and Computers}, 2016, pp. 788--794.

\bibitem{tsianos2012push}
K.~I. Tsianos, S.~Lawlor, and M.~G. Rabbat, ``Push-sum distributed dual
  averaging for convex optimization,'' in \emph{51st IEEE Conference on
  Decision and Control}, Dec 2012, pp. 5453--5458.

\bibitem{xin2018linear}
R.~Xin and U.~A. Khan, ``A linear algorithm for optimization over directed
  graphs with geometric convergence,'' \emph{IEEE Control Systems Letters},
  vol.~2, no.~3, pp. 325--330, 2018.

\bibitem{priolo2014distributed}
A.~Priolo, A.~Gasparri, E.~Montijano, and C.~Sagues, ``A distributed algorithm
  for average consensus on strongly connected weighted digraphs,''
  \emph{Automatica}, vol.~50, no.~3, pp. 946--951, 2014.

\bibitem{cai2012average}
K.~Cai and H.~Ishii, ``Average consensus on general strongly connected
  digraphs,'' \emph{Automatica}, vol.~48, no.~11, pp. 2750--2761, 2012.

\bibitem{nedic2010convergence}
A.~Nedi{\'c} and A.~Ozdaglar, ``Convergence rate for consensus with delays,''
  \emph{Journal of Global Optimization}, vol.~47, no.~3, pp. 437--456, 2010.

\bibitem{bertsekas2015convex}
D.~P. Bertsekas, \emph{Convex Optimization Algorithms}.\hskip 1em plus 0.5em
  minus 0.4em\relax Athena Scientific Belmont, 2015.

\bibitem{Dua2017UCI}
\BIBentryALTinterwordspacing
D.~Dheeru and E.~Karra~Taniskidou, ``{UCI} machine learning repository,'' 2017.
  [Online]. Available: \url{http://archive.ics.uci.edu/ml}
\BIBentrySTDinterwordspacing

\bibitem{xi2018add}
C.~Xi, R.~Xin, and U.~A. Khan, ``Add-opt: Accelerated distributed directed
  optimization,'' \emph{IEEE Transactions on Automatic Control}, vol.~63,
  no.~5, pp. 1329--1339, 2018.

\bibitem{zeng2015extrapush}
J.~Zeng and W.~Yin, ``Extrapush for convex smooth decentralized optimization
  over directed networks,'' \emph{Journal of Computational Mathematics},
  vol.~35, no.~4, pp. 381--394, 2017.

\bibitem{xu2017distributed}
Y.~Xu, T.~Han, K.~Cai, Z.~Lin, G.~Yan, and M.~Fu, ``A distributed algorithm for
  resource allocation over dynamic digraphs,'' \emph{IEEE Transactions on
  Signal Processing}, vol.~65, no.~10, pp. 2600--2612, 2017.

\bibitem{hendrickx2015fundamental}
J.~M. Hendrickx and J.~N. Tsitsiklis, ``Fundamental limitations for anonymous
  distributed systems with broadcast communications,'' in \emph{53rd Annual
  Allerton Conference on Communication, Control, and Computing}.\hskip 1em plus
  0.5em minus 0.4em\relax IEEE, 2015, pp. 9--16.

\bibitem{nedic2016distributed}
A.~Nedic, A.~Olshevsky, and C.~A. Uribe, ``Distributed {G}aussian learning over
  time-varying directed graphs,'' in \emph{50th Asilomar Conference on Signals,
  Systems and Computers}, Nov 2016, pp. 1710--1714.

\bibitem{johnstone2017faster}
P.~R. Johnstone and P.~Moulin, ``Faster subgradient methods for functions with
  h{\"o}lderian growth,'' \emph{Mathematical Programming}, 2019.

\bibitem{nesterov2013introductory}
Y.~Nesterov, \emph{Introductory lectures on convex optimization: A basic
  course}.\hskip 1em plus 0.5em minus 0.4em\relax Springer Science \& Business
  Media, 2013, vol.~87.

\bibitem{uribe2018dual}
C.~A. Uribe, S.~Lee, A.~Gasnikov, and A.~Nedi{\'c}, ``A dual approach for
  optimal algorithms in distributed optimization over networks,'' \emph{arXiv
  preprint arXiv:1809.00710}, 2018.

\bibitem{zhang2019asynchronous}
J.~Zhang and K.~You, ``Asynchronous decentralized optimization in directed
  networks,'' \emph{arXiv preprint arXiv:1901.08215}, 2019.

\end{thebibliography}

\end{document}